\documentclass[11pt]{article}
\usepackage{lmodern}
\usepackage{dsfont}
\usepackage{environ}
\input{definitions.sty}

\newcommand{\Omeg}{\mathrm{\Omega}}
\newcommand{\Omic}{\mathrm{O}}

\newcommand{\Rne}{{\mathbb{R}_{\ge 0}}}
\newcommand{\Rnen}{{\mathbb{R}_{\ge 0}^n}}
\newcommand{\Rnenminusone}{{\mathbb{R}_{\ge 0}^{n-1}}}

\renewcommand{\Pr}{\mathop{\bf Pr\/}}
\newcommand{\E}{\mathop{\mathbb{E}\/}}  %

\newcommand{\eps}{\epsilon}

\newcommand{\calD}{\mathcal{D}}

\newcommand{\ba}{\boldsymbol{a}}

\def\<{\langle}
\def\>{\rangle}

\def\wt{\widetilde}

\def\vec{\bm}

\newcommand{\emailhref}[1]{\href{mailto:#1}{\tt #1}}

\newcommand{\Twin}{{\ensuremath{\mathsf{Twin}}}}

\newcommand{\Viol}{{\ensuremath{\mathsf{Viol}}}}

\newcommand{\sucx}{{\ensuremath{\mathsf{success}}}}

\newcommand{\DSIC}{{\ensuremath{\mathsf{DSIC}}}}
\newcommand{\Mo}{{\ensuremath{\mathsf{M}}}}
\newcommand{\CMo}{{\ensuremath{\mathsf{CM}}}}

\begin{document}

\title{On testing the incentive compatibility \\ of single-parameter allocation mechanisms}%

\author{
 	    \textbf{Jason Milionis} \\
        \small Department of Computer Science \\
 		\small Columbia University \\
 		\small \emailhref{jm@cs.columbia.edu}
   \and
        \textbf{William Pires} \\
        \small Department of Computer Science \\
 		\small Columbia University \\
 		\small \emailhref{williampires@cs.columbia.edu}
}
\date{}
\maketitle

\thispagestyle{empty}

\begin{abstract}
This paper is the first work at the intersection of game theory and property testing, giving algorithms and lower bounds for \emph{efficiently} testing whether an allocation mechanism is incentive compatible (IC). We propose distinguishing whether a mechanism is \emph{$\eps$-far} from being IC, i.e., when it observes many monotonicity ``violations.'' Conceptually, inspired by the literature on Boolean function monotonicity testing, we construct a tester for discrete single-parameter allocation rules. Technically, our work is the first to consider monotonicity testing of \emph{vector-valued} functions on the hypergrid. We give a $\wt{\Omic}(n/\eps)$-query algorithm to test whether a function (representing $n$-player allocation mechanisms) is coordinate-wise monotone versus $\eps$-far from it. We also show a matching lower bound: the class of coordinate-wise monotone vector-valued functions on a Boolean hypercube or hypergrid requires $\wt{\Omeg}(n/\eps)$ queries to test whether it is $\eps$-far from monotonicity, and this holds even if the tester is two-sided and allowed to make adaptive queries. Finally, we extend our upper bound to and give a tester of the same query complexity for pricing functions of allocation mechanisms. This requires overcoming the technical challenge that the
path in function space to the closest IC mechanism may involve interdependent changes to both the price and
the allocation rule.
\end{abstract}

\section{Introduction}
\label{sec:intro}

In the field of algorithmic game theory, mechanism design studies the derivation of rules that induce rational agents to behave in a desired way. A mechanism is a system that specifies the actions, outcomes, and payoffs for each agent, given their preferences and information. A mechanism is said to be incentive compatible if it aligns the agents' individual incentives with the mechanism's social objective, such that no agent has an incentive to deviate from the prescribed strategy.
Carefully considering and reasoning for strategic behavior has played a significant role in the design and analysis of modern algorithms and infrastructure, especially with the advent of the internet \citep{tim_agt}, and has led to the flourishing of the field of mechanism design.

However, nowadays, an increasing array of downstream applications of mechanisms has been operating within environments characterized by partial or obscured transparency, most of the times by virtue of the sheer intricacy and computational intensity of the implemented resource allocation mechanisms.
In this light, and in the era of complex infrastructure and massive computation, a key concern remains unaddressed:
\begin{question}
Is a mechanism that players engage in truthful (even on most inputs) or not?
\end{question}
Simply put, how can participants \emph{efficiently verify} that an allocation mechanism they participate in---and to the internals of which they do not have access to---is actually truthful?
This question forms the focal point of our paper.

A prime example of the significance of probing the incentive compatibility (IC) of black-box
mechanisms—as opposed to trusting that mechanisms are out-of-the-box guaranteed to be IC—
arises from the allegations of anti-competitiveness and non-truthful behavior brought by the \citet{usvgoogle}, illustrating that non-IC mechanisms may have severe real-world implications: distorting market incentives, harming participants, and undermining the economic fairness of allocation mechanisms.
This paper initiates the study of efficiently testing whether a black-box allocation mechanism is truthful to some accuracy $\eps$.
More specifically, our investigation is inspired by and closely follows the extensive literature on property testing.

In property testing, one is interested in designing a fast randomized algorithm (called a tester) to distinguish if an input object has some property $\Pi$, or is far from \textit{any} object with this property.
In particular, in the model we hereby adopt, the algorithm only has query access to the object, and the goal is to develop a tester that makes as few queries as possible, thereby being \textit{efficient}.

A natural question one might ask is: why not test whether the mechanism conforms to \textit{exact} IC, and reject it if it does not?
The reason is that, in the property testing light, such a test is highly inefficient from a computational standpoint---intractable one might say---intuitively because, e.g., in the case of discrete bids, one would have to test every single combination of bids to verify that the mechanism is exactly IC (there might be a single violation of incentive compatibility in a hidden combination of bids that leads to untruthful pricing or allocations); we call such deviations from IC principles, "violations."
Furthermore, besides the utter inefficiency of such a quest, in the key application we have in mind where the allocations are black-box for a good reason (e.g., for problems whose optimal---or even approximate---welfare allocation is computationally taxing to find out), performing all these oracle queries would put a significant strain on outside computational resources.
The query complexity hardness of exactly testing whether an allocation rule $\vec x$ is IC follows roughly from ideas inspired by the traditional principles of property testing literature.
As a matter of fact, as we will later see that our lower bounds establish (see, e.g., \Cref{infthm:second_lb} and set $\eps=2^{-n}$), it is indeed hopeless to aim for the detection of just a \emph{single} violation of incentive compatibility in the whole binary bid space of an auction, as such a quest would necessitate, even under adaptive two-sided testers, exponential query complexity, on the order of $\Omega(2^n)$.

This paper is the first, then, to tackle the problem of \textit{efficiently} verifying whether a given allocation mechanism with black-box access is ($\eps$-)close to some incentive-compatible mechanism, or far from \textit{any} IC mechanism.

\subsection{Description of the Setting}
\label{subsec:intro:setting}

Before we go into our main results and techniques, let us first
describe the setting under which we operate.
Suppose that there are $n$ players. In any mechanism, there is an outcome space $X$ that the mechanism chooses depending on strategic inputs received by players, and the mechanism also directs each player to pay some monetary amount. Players have some utility for (equivalently, valuation of) each outcome of the mechanism; in a single-parameter setting, any user's utility can be fully described by a single, private value $v_i\in\Rne,$ for $i\in [n]$.
An allocation mechanism takes as input from the users bids $b_i$, and each user $i$ obtains an allocation $x_i$ and makes a payment $p_i$ (determined by the mechanism); their utility is then $ u_i = v_i x_i - p_i\,. $
All of the above can equivalently be given in vector format.

A mechanism is called dominant-strategy incentive-compatible (DSIC), if and only if, for all $i\in [n]$ and any bids by the other users $\vec{b}_{-i}\in\Rnenminusone$, the best action for the user is to truthfully report their private value, i.e., for all $b_i\in\Rne$, it holds that $u_i(v_i, \vec{b}_{-i}) \ge u_i(b_i, \vec{b}_{-i})$.
The following is a crisp characterization of all implementable single-parameter DSIC mechanisms, due to Myerson.

\begin{lemma}[\citep{myerson}]
An allocation rule $\vec{x}:\Rnen \to S^n$ for a set $S$ with a total order can be implemented (i.e., there exists an appropriate payment rule that makes the pair of the allocation and payment rule truthful) as a dominant-strategy incentive-compatible (DSIC) mechanism, if and only if
\[ \forall i\in [n], \vec{b}_{-i} \in \Rnenminusone : x_i(b_i, \vec{b}_{-i}) \text{ is monotone in } b_i  \,. \]
We call the above a \emph{coordinate-wise monotone} function. Under the normalization $p_i(0,\vec b_{-i})=0$, the payment rule $\vec{p}:\Rnen\to\Rnen$ must be
$
\forall i\in [n], \vec{b}_{-i} \in \Rnenminusone :
p_i(\vec{b})
=
\int_0^{b_i} s ~dx_i(s; \vec{b}_{-i})
\,,
$
where the integral is a Riemann--Stieltjes integral.
\end{lemma}

Myerson's lemma, therefore, gives an equivalence between incentive compatibility and coordinate-wise monotonicity.
It additionally gives a special form that the normalized payment rule has to obey. Note that the payment rule has this exact form in the case where bids are continuous, and when they are in a discrete space, it corresponds to the maximum payment; we will explore in full detail the structure of these rules in the core of the paper, and will not bother with this distinction for the purposes of this introductory exposition.

The first step towards looking at how to test if a mechanism is incentive-compatible is to understand coordinate-wise monotonicity. Already, the reader might wonder why testing coordinate-wise monotonicity of a function $\vec{x}:\{0,1\}^n \to \{0,1\}^n$ \emph{does not} simply reduce to testing that for each bit $i \in [n]$, the output of $\vec{x}$ is monotone in bit $i$. Here is a counterexample that shows why such an approach would fail. Consider $n=2$ and $\vec{x}$ defined as per $\vec{x}(\vec{00})=\vec{11}, \vec{x}(\vec{10})=\vec{10}, \vec{x}(\vec{01})=\vec{01}$ and $\vec{x}(\vec{11})=\vec{00}$. This function is coordinate-wise monotone, but, if one were to look at only one of the two output bits, the result would not be monotone. Indeed, for $b \in \{1,2\}$, we would always have $\vec{x}(\vec{00})_b =1 > \vec{x}(\vec{11})_b =0$.
That is, if we fix all but coordinate $i$ of the input, the function obtained by looking only at the $i$-th output bit is monotone.

Since the center of our focus in this work will revolve around vector-valued functions on a hypergrid, we will restrict our attention to the discrete case, as an important first step into the venture of property testing and game theory.\footnote{For example, imagine that the bid space has been discretized. Note that such a restriction is not extremely strong, as we have the slack of distinguishing $\eps$-far monotone mechanisms, hence providing natural discretization intervals. However, our focus will not be on such a correspondence.}
At the same time, most common regimes consist of allocation rules that concern combinatorial cases, see, e.g., Roughgarden, 2016 \cite{tim_agt}.
In the discrete case, even after this normalization, a coordinate-wise monotone allocation $\vec{x}$ need not have a unique price function $\vec{p}$ such that the pair $(\vec{x}, \vec{p})$ is DSIC. Indeed, if we allow the prices to be real-valued, there are infinitely many functions $\vec{p}$ that would result in an IC mechanism. However, the behavior of any such $\vec{p}$ is tightly related to the "jumps" in the allocations $\vec{x}$.
Formally, our setting here will be to examine vector-valued functions of the form $\vec{x}:[0:m]^n \to [0:d]^n$ for the allocation rule, and for the joint allocation--price rule, functions of the form $\vec{a}:[0:m]^n \to [0:d]^n \times [0,dm]^n$.

Up to this point, we have almost completely translated the problem to the language of a property testing one. The only definition we are missing is the $\eps$-far part.
For this, we define the distance of a mechanism to IC (in general, we will refer to it as "distance to monotonicity") as the minimum distance of an allocation rule $\vec{x}$ to \textit{any} IC mechanism, i.e., the minimum number of corrections of "violations" in the allocation function needed to be made, in order for the resulting mechanism to be (exactly) IC.
In an analogous way, the distance to monotonicity of a joint allocation--price rule is defined to be the minimum distance to any mechanism that is \textit{jointly} correct, i.e., both allocations and prices are set according to a DSIC rule. In order not to clutter the notation, we will not give the full definition here, but later in \Cref{subsec:main:full}.

\subsection{Main Contributions}
\label{subsec:intro:main}

Our main contributions with this work can be summarized in two key pillars as follows:
\begin{itemize}
\item Conceptually, we propose a novel approach of \textit{efficiently}, approximately verifying the incentive compatibility of allocation mechanisms with few queries. With this, we hope to open up a new area at the intersection of mechanism design and property testing resulting in many fruitful avenues of research.
\item Technically, we make the first known contribution to monotonicity testing of vector-valued Boolean functions, and vector-valued functions on a hypergrid. We show that the class of such coordinate-wise monotone functions is "hard to test" (compared to standard results of the property testing literature) and we derive optimal upper and lower bounds to do so.
\end{itemize}

Before presenting our results, note that a tester can be one or two sided, and adaptive or non-adaptive. We provide definitions of these in \Cref{sec:prelims} and sketch them here. Roughly speaking a tester $\mathcal{A}$ for a property $\Pi$ \footnote{$\Pi$ should be thought as a set of function} gets as input query access to a function $f$ and must have the following properties: if $f$ is "far" from any function in $\Pi$, the algorithm should reject with high probability. If $\mathcal{A}$ is one-sided, it must \emph{always} accept when $f \in \Pi$. If $\mathcal{A}$ is two-sided, it must accept $f$ with \emph{high probability} when $f \in \Pi$. If $\mathcal{A}$ is non-adaptive, the algorithm makes all its queries at once in parallel before deciding to accept or reject. However if $\mathcal{A}$ is adaptive, it's $i$-th query can depend on all the previous queries as well the answers received for them. 

We will now present informal versions of our main theorems, whose full versions can be found in \Cref{theorem: second lower bound,theorem: grid tester,thm:ub,thm:ub_amazing}, respectively. We first look at the complexity of testing coordinate-wise monotonicity.
\begin{inftheorem}
\label{infthm:second_lb}
Any two-sided adaptive tester for coordinate-wise monotonicity must have query complexity $\wt{\Omeg}(n/\epsilon)$.
\end{inftheorem}

\begin{inftheorem}
\label{infthm:ub1}
There exists a one-sided non-adaptive tester for coordinate-wise monotonicity with query complexity $\Omic(n \log(m) /\epsilon)$.
\end{inftheorem}

Using the insights we get from testing coordinate-wise monotonicity, we tackle the more challenging task of testing incentive-compatibility. We then give two testers for DSIC mechanisms. First, we show that if computational complexity is not a consideration---but rather oracle queries are---then an algorithm of ours obtains a near-optimal (because of our lower bound in \Cref{infthm:second_lb}) sample complexity (in \Cref{infthm:ub3}). However, with polynomial computational complexity, we identify a tester that does not asymptotically achieve the near-optimal sample complexity shown by \Cref{infthm:second_lb} (in \Cref{infthm:ub2}). We leave the issue of potentially closing this gap (or providing a formal separation of the two regimes) as a very interesting future direction of study in the new area of literature that we initiate.

\begin{inftheorem}
\label{infthm:ub2}
There exists a one-sided tester for DSIC mechanisms (including prices) with query complexity
$\Tilde{\Omic}\left( \frac{n}{\epsilon} \min\left(m, d\right) \right)$. The running time of the tester is $\Tilde{\Omic}(n^2d/\epsilon)$. 
\end{inftheorem}

\begin{inftheorem}
\label{infthm:ub3}
There exists a one-sided tester for DSIC mechanisms (including prices) with query complexity
$\Tilde{\Omic}(n/\epsilon)$. The running time is $\Tilde{\Omic}(n^2/\epsilon)+\Tilde{\Omic}(n/\epsilon)d^{\Omic(m)}$.
\end{inftheorem}

\subsection{Techniques}
\label{subsec:intro:techniques}

\subsubsection{Lower bounds.} Our lower bounds concern the case of coordinate-wise monotonicity testing for vector-valued Boolean functions (i.e., a hypercube), and since the case of a hypergrid contains as a special case that of the hypercube, these lower bounds also apply then, and also even if we also test prices (which we do in the above general algorithm we gave).
The tester we will analyze in \Cref{subsubsec:intro:ub} matches the lower bound, hence it is optimal to test both allocations and prices when prices are derived from Boolean allocations, and even in the most general case of completely generic discrete allocation mechanisms (hypergrid allocations), up to log factors.

We show that testing whether a function is coordinate-wise monotone or $\epsilon$-far from it is hard, even for a two-sided adaptive tester.
Before rising to the full generality of this claim, let us first consider (as an instructive warm-up) a simpler case of our lower bound against one-sided non-adaptive testers, to give a better intuition as to why coordinate-wise monotonicity might be a hard-to-test property.\footnote{The rest of the proof that we develop in the main text goes through a classic lower bound technique against one-sided testers \citep[see, e.g.,][]{fisher}.}
\begin{lemma}\label{infthm:first_lb}
For every $\eps \le 1/2$, any one-sided non-adaptive $\eps$-tester for coordinate-wise monotonicity must have query complexity $\Omega(n)$.
\end{lemma}
The key point is that, to show the hardness, we rely on the following lemma, whose formal version will be presented in \Cref{lemma:key}, and roughly arises from a characterization of coordinate-wise monotonicity:
\begin{lemma}[informal]
A one-sided $\eps$-tester for coordinate-wise monotonicity cannot reject a function $\vec{x}: \{0,1\}^n \to \{0,1\}^n$, unless it queries the two end points of an {\em edge} violating coordinate-wise monotonicity.
\end{lemma}

In order to proceed, we extend the above lemma (presented in \Cref{theorem: first lower bound}) to one that shows that, even with two-sided adaptive testers, querying an entire set of points does not help unless these specifically include a collection of at least $\Omega(n)$ end points of {\em edges}, and we prove our stronger lower bound, referenced above as \Cref{infthm:second_lb}, and formally shown in \Cref{theorem: second lower bound}.

The proof of our $\wt{\Omega}(n / \epsilon)$ lower bound is inspired by the technique used by \citet{CHOCKLER2004301}, who proved a lower bound of $\Omega(k)$ against two-sided adaptive testing of $k$-juntas.\footnote{A $k$-junta is a function on the hypercube that depends on at most $k$ variables.}

More specifically, we  define two randomized processes, $P_0$ and $P_1$, which reply to the adaptive queries of the tester while building a function $\Vec{x} : \{0,1\}^n \rightarrow \{0,1\}^n$. The function built by $P_0$, during this procedure, is $\epsilon$-far from coordinate-wise monotonicity with very high probability while the function built by $P_1$ is always coordinate-wise monotone. In particular, $P_0$ starts by sampling a small random set $S$ of edges along which it places violations to coordinate wise monotonicity as it replies to the algorithm's queries. 

Therefore, a two-sided tester for coordinate-wise monotonicity should reject with high probability when interacting with $P_0$ but accept with high probability when interacting with $P_1$. However, we argue that an algorithm can't distinguish $P_0$ from $P_1$ as long as it hasn't queried the two end points of an edge in $S$.
Unlike \citet{CHOCKLER2004301}, we are able to make our lower bound depend on $\epsilon$. There are two reasons for this: first, the way we define the behavior of $P_0$ depends on $\epsilon$ and we argue in \Cref{lemma: points to edges} that if an algorithm only makes $q$ queries, it has queried the two end points of at most $q\log_2 q$ edges. Thus, if the algorithm makes few queries, it's unlikely to hit both end points of an edge in $S$ when interacting with $P_0$.%

\subsubsection{Upper bounds.}
\label{subsubsec:intro:ub}
We start by demonstrating the key ideas to get our one-sided non-adaptive tester for coordinate-wise monotonicity. 

\begin{theorem}\label{infthm: grid tester}
There exists a one-sided non-adaptive tester with query complexity $\Omic(n \log(m)/\epsilon)$ that tests if a function $\Vec{x} : [0:m]^n \to [0:d]^n$ is coordinate-wise monotone or $\epsilon$-far from monotonicity.
\end{theorem}

In our tester, we turn to the analysis of \citet{oded_improved}, where the authors were interested in testing monotonicity of functions $f : [0:m]^n \rightarrow \Xi$, where $\Xi$ is an {\em ordered} set, which does not apply for our setting.
Specifically, given the function $\vec{x} : [0:m]^n \rightarrow [0:d]^n$, we decompose the hypergrid into "lines". That is, for $i\in [n], \alpha \in [0:m]^{i-1}, \beta \in [0:m]^{n-i}$, we define the line $x_{\alpha,i, \beta} :  [0:m] \rightarrow [0:d]$, as the restriction $x_{\alpha,i, \beta}(b)=\vec{x}(\alpha,b, \beta)_i$ for all $b \in [0:m]$.

We show that sampling a random "line" of the hypergrid gives a function that, in expectation, is at least $\epsilon/n$ far from monotonicity. We combine this with a result of \citet{oded_improved} who gave a distribution $\mathcal{D}$ such that if $f : [0:m] \rightarrow [0:d]$ is far from monotonicity, a pair $(p,l)$ sampled from the $\mathcal{D}$ witnesses a violation to monotonicity with high probability. Thus, if $\vec{x}$ if far from being coordinate-wise monotone, the following test should reject with "good" probability: Sample a line $x_{\alpha,i,\beta}$ of $\vec{x}$ uniformly at random, draw $(p,l) \sim \mathcal{D}$ and reject if $x_{\alpha,i,\beta}(p)>x_{\alpha,i,\beta}(l)$. Also observe that this test never rejects when $\vec{x}$ is coordinate wise monotone. Our tester then simply repeats the above test $O(n\log(m)/\epsilon)$ times.

Finally, we give two one-sided testers to test whether a pair $(\vec{x}, \vec{p})$, of an allocation function and a price function form a $\DSIC$ mechanism.

\begin{theorem}\label{ifnthm: price tester}
\hfill
    \begin{itemize}
        \item There exists a one-sided tester with query complexity $\Omic( \frac{n}{\epsilon} \log(2n/\epsilon) ( \min\{m,d\log(m)\} + \log(2n/\epsilon)) )$ to test if a pair $(\Vec{x}, \Vec{p}) : [0:m]^n \to [0:d]^n \times [0,dm]^n$ implements a $\DSIC$ mechanism or is $\epsilon$-far from any $\DSIC$ mechanism. This tester runs in time $\Tilde{\Omic}(n^2d/\epsilon)$.
        \item There exists a one-sided tester with query complexity $\Omic( \frac{n}{\epsilon}\log(m)\log^2(2n/\epsilon))$ for the same task. This tester runs in time $\Tilde{\Omic}(n^2/\epsilon)+\Tilde{\Omic}(n/\epsilon)d^{\Omic(m)}$.
    \end{itemize}
\end{theorem}

The tester of \Cref{infthm: grid tester} is optimal (up to log factors) for testing coordinate-wise monotonicity.
Hence, we use a similar idea for the joint testing of allocation rules and prices. As a reminder, given a function $\vec{a}=(\vec{x}, \vec{p}) : [0:m]^n \rightarrow [0:d]^n \times [0,dm]^n$, we decompose it into "lines". As a corollary of the $\DSIC$ definition, we derive what it means for the "line" $a=(x,p) : [0:m] \rightarrow [0:d] \times [0,dm]$ to obey $\DSIC$, and show that again, if $\vec{a}$ is $\epsilon$ far from being $\DSIC$, then a randomly sampled line $a_{\alpha, i, \beta}$ is at distance $\epsilon/n$ from $\DSIC$ in expectation. 
Assume we have an algorithm that test if function $a:(x,p)$ is $\DSIC$ or $\epsilon$-far from being $\DSIC$ (we call such an algorithm a line tester). Then we can test if $\ba$ is $\DSIC$ using the work investment strategy of \citet{goldreich:LIPIcs.APPROX-RANDOM.2014.704}. Roughly speaking, we try different distance parameters $\epsilon'$: For each $\epsilon'$ we sample $\tilde{O}(n/(\epsilon \cdot \epsilon'))$ lines of $\ba$ and run our line tester on each them with distance parameter $\epsilon'$. If a line is rejected by the tester, we reject $\Vec{a}$, otherwise we accept. We can show that by picking the different choices of $\epsilon'$ appropriately, with high probability $\Vec{a}$ is rejected if it is $\epsilon$ far from $\DSIC$. 

In particular, we give two different line testers, which can then be used to get the two different algorithm of \Cref{ifnthm: price tester}.

Our first tester, which is a "learner", works by finding all the "jumps" in the allocation function $x$, and then guessing a $\DSIC$ function $c$ based on the value of $x$ and $p$ at these points. If $a$ is $\DSIC$ this function is $a$ itself, otherwise $a$ must be far from $c$ so $a$ and $c$ must disagree on a randomly chosen point with high probability.

The second tester is more involved. We consider a tree binary tree $T_a$ of height $\Omic(\log(m))$. The root of the tree is labeled with the interval $[0:m]$. Then, given a node $[s:t]$, its left child is with the first half of the interval while the right child is labeled with the second half, the leafs are labeled with $[b:b]$ for each $b \in [0:m]$. Each node $[v:w]$ in the tree is associated with a set $\mathcal{P}_{[v:w]}$ of $\DSIC$ functions. Roughly speaking $\mathcal{P}_{[v:w]}$ corresponds to the set of $\DSIC$ functions which agree with $a(x)$ for all $x \in \big\{s,t \mid [s:t] \text{ is an interval on the path from the root of $T_a$ to $[v:w]$}\big\}.$

As such, given $b \in [0:m]$, we can check with $\Omic(\log(m))$ queries if this set is empty. Finally, If $a$ is $\DSIC$ then for all $b$ we have  $\mathcal{P}_{[b:b]}\neq \emptyset$ (since $a$ itself if always in $\mathcal{P}_{[b:b]}$), otherwise with high probability $\mathcal{P}_{[b:b]}= \emptyset$ on a randomly chosen $b$.

Our first line tester has query complexity $O(\min\{m,d\log(m)\}+1/\epsilon)$ while our second tester uses only $O(\log(m)/\epsilon)$ queries. Computing if $\mathcal{P}_{[b:b]}$ is empty is non trivial, and takes time $d^{O(m)}$: we search over all $(d+1)^{m+1}$ possible choices of allocation $x$, to see if there is some way to assign prices while being consistent with $a$ on the needed points (once $x$ is fixed, this is done using a linear program). This explains the difference in query complexity and running time between the two results of \Cref{ifnthm: price tester}.

\paragraph{Distance and distribution.}
So far, we haven't mentioned how the distance from $\Vec{a} : [0:m]^n \to [0:d]^n \times [0,dm]^n$ to being $\DSIC$ is computed. One could consider the uniform distribution, that is we are interested in minimum fraction of points in $[0:m]^n$ where one needs to change $\Vec{a}$ so that the resulting function is $\DSIC$. However, when thinking about mechanisms, it seems unlikely that the users' valuations are distributed uniformly. Therefore, when testing if $\vec{a}$ is $\DSIC$, we will actually consider its distance to being $\DSIC$ under an unknown distribution $\mathcal{D}$ over $[0:m]^n$. That is, given $\Vec{a}$ what is the minimum over all $\DSIC$ functions $\Vec{c}$ of $\Pr_{\vec{b} \sim \mathcal{D}}\left[ \vec{a}(\Vec{b}) \neq \vec{c}(\Vec{b}) \right]$.

In particular, we will consider that our testers have \emph{conditional sample access} to $\mathcal{D}$. This means we can sample $\Vec{b} \sim \mathcal{D}$ but also that for each user $i$, we can fix a vector a vector $\vec{b_{-i}}$ and ask for a sample $\Vec{b} \sim \mathcal{D}$ conditioned on the bids of users $1,\ldots, i-1, i+1, \ldots, n$ being $\vec{b_{-i}}$. This is equivalent to sampling the bid $b \in [0:m]$ of user $i$, conditioned on the bids of the other users being $\vec{b_{-i}}$. 

If one assumes bidders act independently, then  $\mathcal{D}$ is a product distribution $\mu_1 \times \ldots \times \mu_n$ where each $\mu_i$ is a distribution over $[0:m]$. Note that in such a case, we do not need conditional sampling access to $\mathcal{D}$. Indeed for any $i \in [n]$ the distribution of bids of user $i$ conditioned on the bids of the other users being some fixed vector $\vec{b_{-i}}$ follows $\mu_i$. Hence, whenever we need a conditional sample for some $i \in [n]$, we can instead sample $\Vec{b}$ from $\mathcal{D}$ and then look at $\vec{b}_i$.

\paragraph{A remark about $\eps$-IC testing.}
In mechanism design, we sometimes refer to approximately truthful mechanisms, where the desideratum (in the worst-case, distribution-independent version) is for the maximum utility gain that any single player can get from deviating from the truthful strategy to be at most $\eps$.
However, from our results, it is not straightforward to test for $\eps$-IC, and we leave this direction as a very interesting future avenue of research.
The following simple example shows why we cannot hope to get $\eps$-IC consequences from $\eps$-far monotonicity (or vice-versa) in a meaningful way: it is possible to create instances that are far from being $\eps$-IC, but with a single change of a price, can be made exactly DSIC; the reverse is also true, and intuitively, a function may be very far from \emph{any single} monotone allocation rule, but still be close to being approximately IC from the point of view of any single player's deviation in their strategy space.
As $\eps$-IC is out of scope for the present work, we will not elaborate further on that.

\subsection{Related Work}
\label{subsec:litrev}

Hereby, we offer a comprehensive review of various parts of the relevant literature.

\citet{estimatingepsICBalcan19} use samples from agents' type distributions in a Bayesian setting to bound the \textit{average} utility gain from strategic deviations over the distribution of other players' types, and, in a sample complexity manner, estimate the (Bayesian) $\eps$-IC of a mechanism.
The distributional assumption assists their quest.
In contrast, our work seeks to uniformly and unconditionally test the distance to monotonicity;
see the discussion in \Cref{subsec:intro:techniques} on the relation of that to \textit{worst-case} utility gains from deviating.
\citet{ab1_sergei,ab2} focus on the perspective of systematic testing from the viewpoint of a single bidder in the specific case of a second-price auction with reserves, and derive A/B testing setups for this format, utilizing simple bid perturbations.
Our work considers a significantly larger standpoint of allocation mechanisms, and thus exact testing is computationally inefficient to achieve, as highlighted in \Cref{sec:intro}.
Moreover, their work does not concern or relate to property testing, which is our predominant focus in the current paper.

The above work relates to a more general literature on learning from auctions, mostly with the goal of figuring out (optimal or almost optimal) reserve prices \citep{revmaxsinglesample10,chawlahartline2010,cole_roughgarden,makingmostofsamples15,pseudodim,fieldguidepersprice16,minregreserves16}. Such work could be used to verify if the reserve price is approximately what would be expected, but we are in a much more generic setting with no guarantees that the allocation observes such a format, as described above.

Another line of work centers around designing verifiably correct auctions from the point of the auctioneer who might be interested in malicious deviations \citep{PARKES2008294,credible}; such work is only tangentially relevant to our focus here, since it is usually not completely general for allocation rules, but most importantly, their objective lies on trying to make following the auction be a dominant strategy of the auctioneer, hence assuming a particular utility structure for the auctioneer, who nonetheless might have other reasons to deviate from the auction's specification.
We, in contrast, do not assume anything about the reason for the mechanism's deviation from IC, i.e., about the source of (reason for) violations.
A now-mature line of research is dedicated to bypassing the impossibility of designing credible IC mechanisms \citep{credible} by using cryptographic commitments and multi-round auction mechanisms \citep{jason_nft_auctions,chitraferreirakulkarni,ferreiraweinberg}.

Our work also contributes to the literature on trust and transparency in market design.
\citet{canetti2023zeroknowledgemechanisms} study the possibility of participants verifying whether a mechanism is DSIC, and doing so without revealing that the mechanism itself (in what is called the ``zero-knowledge'' approach). Unlike our work, \citet{canetti2023zeroknowledgemechanisms} abstract away from quantifying the number of queries it takes to test DSIC, or the computational complexity of testing algorithms.
Relatedly, \citet{gonczthomas2024} study the question of how many bits of information it takes for participants to verify that correct allocation mechanism has been used.
\citet{grigoryan_auditability2023} study the question of market participants detecting whether the designer runs the mechanism as promised; unlike \citep{credible}, and similarly to our work, they are concerned with general allocation problems and arbitrary reasons for deviations.

The other side of the literature that our work builds upon and extends is that of property testing of Boolean functions and hypergrids. Starting with \citet{FirstMonoTester}, extensive work has been done in that area \citep{10.1145/1007352.1007414, Xi_mono_lb, chen2014booleanfunctionmonotonicitytesting, belovs2015polynomiallowerboundtesting, best_mono_tester, chen2017talagrand, oded_improved, fisher, BestGridLowerBound, chakrabarty2013OptUpperBound,black2023d12o1}, culminating in a non-adaptive algorithm with $\wt{\Omic}(\sqrt{n}/\epsilon^2)$ query complexity by \citet{best_mono_tester}. To this date, there are no known adaptive testers achieving better query complexity. Non-adaptive testers have a matching (up to $\log$ factors) lower bound of $\wt{\Omeg}(\sqrt n)$ queries when $\epsilon$ is a constant \citep{chen2017talagrand}. Nevertheless, for adaptive testers, the gap is still open, since the best lower bound is $\wt{\Omeg}(n^{1/3})$ \citep{chen2017talagrand}.
On the front of property testing for vector-valued Boolean functions, some works \citep{quantum1, quantum2} on quantum computing have been interested in various flavors of testing, e.g., for linearity; nonetheless, none of these works concern monotonicity testing (our sole focus), and other quantum testing papers \citep{belovs2015quantum} only consider the case of univariate functions on the Boolean hypercube.
As mentioned before, when testing if a function $\Vec{a}$ is $\DSIC$ we will consider the distance with respect to unknown distribution $\mathcal{D}$ and we assume our tester can has conditional sample access to $\mathcal{D}$. This model is very similar to the distribution free model of property testing first investigated by \citep{10.1145/285055.285060}. In this model one has query access to a function $f$ and sample access to a an unknown distribution $\mathcal{D}$ (but one cannot ask for conditional samples). The goal is to test whether $f$ has some property or if for any function $g$ with the property we have $\Pr_{x \sim \mathcal{D}}[f(x) \neq g(x)] \geq \epsilon$.

Lastly, in general, results in high-dimensional hypergrid and hypercube settings has been notoriously hard to obtain. Inequalities involving vector-valued functions on the Boolean hypercube have recently been given by \citet{beltran2023sharp,eskenazis_polynomial_2020}, but these have no algorithmic or testing consequences known to this date, like the ones we focus on, nor do they easily translate to the hypergrid.

\section{Preliminaries and Definitions}
\label{sec:prelims}

Unless stated otherwise, $m,n,$ and $d$ are positive integers. We denote by $[m]$ the set $\{1,\ldots,m\}$ and by $[0:m]$ the set of integers $\{0,1,\ldots,m\}$.
In this work we will be interested in vector-valued (or multi-output) functions, meaning their output is a vector. Some examples that we will use in this work are $f:\{0,1\}^n \to \{0,1\}^n$, $f:[0:m]^n \to [0:d]^n$, or $f:[0:m]^n \to [0:d]^n \times [0,dm]^n$.

Given $\alpha \in [0:m]^{i-1}, b \in [0:m]$ and $\beta \in [0:m]^{n-i}$, we will often denote by $\alpha,b,\beta$ their concatenation. That is $\alpha,b,\beta$ is the string that agrees with $\alpha$ on coordinates $j \leq i-1$, has $i$-th coordinate being $b$, and agrees with $\beta$ on coordinates $j\geq i+1$.

We proceed to give some basic definitions that we will later use.
\begin{definition}[Monotonicity]
    Let $\Sigma$ and $\Xi$ be totally ordered sets. A function $x : \Sigma \to \Xi$ is monotone if $x(b)\leq_\Xi x(b')$ whenever $b\leq_\Sigma b'$.
\end{definition}

\begin{definition}[Distance between two functions]
Given two functions $x,y: S \to T$, we define the distance between $x$ and $y$ as
\[
d(x,y) \triangleq \frac{1}{|S|}\left|\{s \in S\;:\; x(s) \neq y(s) \}\right|
\,.
\]
\end{definition}

\begin{definition}[Distance to monotonicity]
\label{def:dist_mono}
    Let $\Sigma$ and $\Xi$ be finite totally ordered sets and let $x:\Sigma\to\Xi$. The distance of $x$ to monotonicity is
    \[d_\Mo(x) \triangleq \min_{y} d(x,y) \,,\]
    where the minimum ranges over all monotone functions $y:\Sigma\to\Xi$.
\end{definition}

\begin{definition}[Coordinate-wise monotonicity and distance]
\label{def:dist_cmo}
    A function $\vec{x}:[0:m]^n\to[0:d]^n$ is coordinate-wise monotone if, for every $i\in[n]$, $\alpha\in[0:m]^{i-1}$, and $\beta\in[0:m]^{n-i}$, the scalar function
    \[
        b\longmapsto \vec{x}(\alpha,b,\beta)_i
    \]
    is monotone. Its distance to coordinate-wise monotonicity is
    \[
        d_\CMo(\vec{x})\triangleq \min_{\vec{y}}d(\vec{x},\vec{y}),
    \]
    where the minimum ranges over all coordinate-wise monotone functions $\vec{y}:[0:m]^n\to[0:d]^n$.
\end{definition}

Likewise, the distance to monotonicity of a joint allocation and price rule is defined to be the minimum distance to any mechanism that is \textit{jointly} correct, i.e., both allocations and prices are set correctly. So that we not clutter the notation, we will not give the full definition here, but later in \Cref{subsec:main:full}.
    
We will consider two different types of testers, two-sided and one-sided testers.
\begin{definition}[Two sided property tester]
    Consider a property $\Pi \subseteq \{f \;|\; f : S \to T\}$. A tester for $\Pi$ is a randomized algorithm $\mathcal{A}$ that, given a parameter $0<\epsilon\leq1$ and query access to a function $x : S \to T$, satisfies the following conditions:
    \begin{enumerate}
        \item $\mathcal{A}$ accepts a function from $\Pi$ with high probability. 
        That is, if $x \in \Pi$, then \[\Pr[\mathcal{A}(\epsilon, x) = 1] \geq 2/3.\]
        \item $\mathcal{A}$ rejects a function that is far from $\Pi$ with high probability. That is, if $d(x,y) \geq \epsilon$ for all $y \in \Pi$, then \[\Pr[\mathcal{A}(\epsilon, x) = 0] \geq 2/3.\]
    \end{enumerate}
    In the above, the probabilities are taken over the randomness used by $\mathcal{A}$.
\end{definition}

\begin{definition}[One-sided property tester]
    The definition is similar to that of a two-sided tester, except the first condition becomes: if $x \in \Pi$ the tester must \emph{always} accept $x$. That is, if $x \in \Pi$, then $\Pr[\mathcal{A}(\epsilon, x) = 1] =1$.
\end{definition}
In particular, a one-sided tester cannot reject a function $x$, unless, based on the queries it made it can infer that $x \not \in \Pi$. Qualitatively, the difference between one-sided and two-sided testers might appear to be minor, but in fact, there are strong separations between achievable outcomes for the two kinds of testers, see, e.g., the reference by Roughgarden, 2021 \cite{tim_BWCA}.

Another distinction is to consider adaptive and non-adaptive testers. Letting $b_k$ denote the $k$-th queried made by the tester, the way an adaptive tester $\mathcal{A}$ picks its $k$-th query can be viewed as a randomized mapping from $\left[(b_1, x(b_1)), \ldots, (b_{k-1}, x(b_{k-1})) \right]$ to $S$. On the other hand, if $\mathcal{A}$ is non-adaptive, then the choice of $b_k$ isn't allowed to depend on the answer to the queries $b_1, \ldots, b_{k-1}$. So, we can view a non-adaptive tester as performing all its queries at once, then receiving $x(b_k)$ for each of them, and finally choosing to accept $x$ or not.

\section{Main Results}
\label{sec:main_results}

\subsection{Lower bounds for coordinate-wise monotonicity}
We begin our study by considering the case of binary bids and binary allocations, where we show a lower bound of $\wt{\Omeg}(n / \epsilon)$ against two-sided adaptive testers. To give better intuition for the reader as to why testing coordinate-wise monotonocity might be harder (as in, requiring more queries) than might be expected from standard property testing results, we first consider the simpler case of an $\Omeg(n)$ lower bound for one-sided non-adaptive testers.

\begin{definition}[$i$-twin]
    Let $\vec{b} \in \{0,1\}^n$ and let $i \in [n]$. The $i$-twin of $\vec{b}$, denoted $\vec{b^{(i)}}$, is the unique element of $\{0,1\}^n$ that differs from $\vec{b}$ only in the $i$-th coordinate, i.e., it has that coordinate flipped.

    Given $i \in [n]$, we define $\Twin(i) \triangleq \{ (\vec{b},\Vec{b^{(i)}}) \;: \; \vec{b}_i = 0\}$.
    So $\Twin(i)$ is the set edge of the
hypercube in the $i$-th direction, where the first element of the edge is always the one with a $i$-th bit
of 0. Finally, given $S \subseteq [n]$, we define $\Twin(S)\triangleq\bigcup\limits_{i \in S}{ \Twin(i) }$.
\end{definition}

\begin{definition}[$\Viol(i)$]
    Let $\vec{x}:\{0,1\}^n \to \{0,1\}^n$ and $i \in [n]$. We define
    \[
    \Viol(i)_{\vec{x}} \triangleq \left\{ (\vec{b},\vec{b^{(i)}}) \;: \;  (\vec{b},\vec{b^{(i)}}) \in  \Twin(i) \wedge \vec{x}(\vec{b})_i > \vec{x}(\vec{b^{(i)}})_i \right\}
    \,.
    \]

    Given $S \subseteq [n]$, we define $\Viol(S)_{\vec{x}} \triangleq \bigcup\limits_{i \in S}{ \Viol(i) }$.
\end{definition}
When it is clear from context, we will drop the subscript $\vec{x}$ from $\Viol(S)_{\vec{x}}$. We will often refer to an edge $(\vec{b},\vec{b^{(i)}})$ of the hypercube as \emph{violating} if $\vec{b}_i=0$ but $\vec{x}(\vec{b})_i > \vec{x}(\vec{b^{(i)}})_i$, where $i$ is the unique coordinate that $\vec{b}, \vec{b^{(i)}}$ differ on.

Using the definition of coordinate-wise monotonicity, the following lemma follows easily.
\begin{lemma}\label{def : cwm cube}
    Let $\vec{x}: \{0,1\}^n \to \{0,1\}^n$, $S \subseteq [n]$. We say that $\vec{x}$ is $S$-monotone if and only if $\Viol(S) = \emptyset$. $\vec{x}$ is coordinate-wise monotone if and only if it is $[n]$-monotone.
\end{lemma}

We use the distance $d_{\CMo}$ from \Cref{def:dist_cmo}; a change at one input counts once even if several output coordinates change there.

\subsubsection{One-sided non-adaptive lower bound.\\}

As a warm-up, we first present our $\Omeg(n)$ lower bound against one-sided, non-adaptive testers.
\begin{theorem}\label{theorem: first lower bound}
    Let $0<\epsilon \leq 1/2$ and $\vec{x}:\{0,1\}^n \to \{0,1\}^n$. Let $\mathcal{A}$ be a one-sided non-adaptive tester that rejects $\vec{x}$ with probability at least $2/3$ when $d_\CMo(\vec{x}) \geq \epsilon$. Then, $\mathcal{A}$ must have query complexity $\Omega(n)$.
\end{theorem}

The following lemma which shows that coordinate-wise monotonicity is a very local property, will be key in our lower bound. This is in contrast with monotonicity testing, where a one-sided algorithm can reject  $f : \{0,1\}^n \rightarrow \{0,1\}$ as long as it finds $ a \prec b$ with $f(a) > f(b)$, even if $a,b$ are not neighbors. 

\begin{lemma}\label{lemma:key}
A one-sided tester for coordinate-wise monotonicity, can't reject $\vec{x}$ unless it queried the two endpoints of an edge in $\Viol([n])$.
\end{lemma}

\begin{proof}{Proof}
Consider a one-sided non-adaptive tester $\mathcal{A}$ for coordinate-wise monotonicity. Say that $\mathcal{A}$ made $q$ queries to an unknown function $\vec{x}:\{0,1\}^n \to \{0,1\}^n$, and has learnt a set of pairs $(\vec{b_1}, \vec{x}(\vec{b_1})), \ldots (\vec{b_q}, \vec{x}(\vec{b _q}))$. Furthermore, assume that the algorithm hasn't queried both end points of any edges in $\Viol([n])_{\vec{x}}$. Then, we can define function $\vec{y}:  \{0,1\}^n \to \{0,1\}^n$ with the following properties: 
\begin{itemize}
    \item $d_\CMo(\vec{y})=0$.
    \item For all $j \in [q]$, we have $\vec{y}(\vec{b_j})=\vec{x}(\vec{b_j})$. That is, $\vec{y}$ is consistent with $\vec{x}$ on points queried by $\mathcal{A}$.
\end{itemize}

Assuming such a function exists, since $\mathcal{A}$ is one-sided, it cannot reject $\vec{x}$ since it could be coordinate-wise monotone, which proves the lemma. We define $\vec{y}$ on all unqueried points $\vec{b}$ as $\vec{y}(\vec{b})=\vec {b}$ and for all queried points $\vec{b_j}$, we define $\vec{y}(\vec{b_j})=\vec{x}(\vec{b_j})$.

We now show that $\vec{y}$ is coordinate-wise monotone. Let $i \in [n]$. And let $(\vec{b}, \vec{b'}) \in \Twin(i)$ be an edge in the $i$-th direction, we claim we always have $\vec{y}(\vec{b}) \leq \vec{y}(\vec{b'})$. We consider several cases:
\begin{itemize}
    \item Case 1: Both $\vec{b}, \vec{b'}$ were queried by $\mathcal{A}$. Then, $\vec y(\vec{b})=\vec{x}(\vec{b})$ and $\vec y (\vec{b}')=\vec{x}(\vec{b}')$, but  $\mathcal{A}$ didn't query any violating edge, so we must have $\vec y(\vec{b})_i=\vec{x}(\vec{b})_i \leq \vec{x} (\vec{b}')_i=\vec y(\vec{b}')_i$. So there is no violation along that edge.
    \item Case 2: Neither $\vec{b}, \vec{b'}$ were queried by $\mathcal{A}$. Then, we defined $\vec y(\vec{b})_i=\vec{b}_i=0$ and $\vec y(\vec{b'})_i=\vec{b'}_i=1$, so there is no violation along the edge.
    \item Case 3: $\vec{b}$ was queried by $\mathcal{A}$ but not $\vec{b'}$. Then, we  have $\vec{y}(\vec{b'})_i=\vec{b'}_i=1$. So, no matter what $\vec{x}(\vec{b})$ is, we must have $\vec{y}(\vec{b})_i= \vec{x}(\vec{b})_i \leq 1 = \vec{y}(\vec{b'})_i$. Hence, there is no violation along that edge.
    \item Case 4: $\vec{b'}$ was queried by $\mathcal{A}$ but not $\vec{b}$. Then, we have $\vec{y}(\vec{b})_i=\vec{b}_i=0$. So, no matter what $\vec{x}(\vec{b'})$ is, we must have $\vec{y}(\vec{b'})_i= \vec{x} (\vec{b'})_i \geq 0 = \vec{y}(\vec{b})_i$. Hence, there is no violation along that edge.
\end{itemize}
Thus, $\vec{y}$ is coordinate-wise monotone, as claimed.\Halmos
\end{proof}

We will also use the following lemma by \citet{CHOCKLER2004301}.
\begin{lemma}[\cite{CHOCKLER2004301}]\label{lemma: set size to edges }
    Given a set $Q \subseteq \{0,1\}^n$, we have 
    $$ |\{i \;: \; (\vec{b},\vec{b'}) \in Q \times Q \text{ and } (\vec{b},\vec{b'}) \in \Twin(i) \}| \leq |Q|-1.$$

    In particular, given a set $Q$ of $q$ points from the hypercube, there are at most $q-1$ different $i$ such that $\vec b$ and $\vec{b^{(i)}}$ were queried.
\end{lemma}

We can now prove the lower bound of \Cref{theorem: first lower bound}.

\begin{proof}{Proof}
Let $\mathcal{A}$ be a one-sided non-adaptive tester for coordinate-wise monotonicity with query complexity $q$. Given unknown function $\vec{x}:\{0,1\}^n \to \{0,1\}^n$, $\mathcal{A}$ must reject $\vec{x}$ with probability at least $2/3$ if $\vec{x}$ is $\epsilon$ far from coordinate-wise monotone and if $\vec{x}$ is coordinate-wise monotone it must always accept. We will show $q \geq \frac{2}{3}n$. 

For every $i \in [n]$, we define function $\vec{x_i}$ as follow: For all $\vec{b} \in \{0,1\}^n$:
\begin{itemize}
    \item For all $j$, $\vec{x_i}(\vec{b})_j = \vec{b}_j$. 
    \item $\vec{x_i}(\vec{b})_i = 1-\vec{b}_i$. 
\end{itemize}
We have the following easy to see properties about $\vec{x_i}$: All violations are along edges in the $i$-th direction ($\Viol([n])_{\vec{x_i}} = \Viol(i)_{\vec{x_i}}$), and $\vec{x_{i}}$ is at distance $1/2$ from coordinate-wise monotonicity.

By \Cref{lemma:key}, we know $\mathcal{A}$ can reject $\vec{x_{i}}$ if and only if it queries the two end points of some edge in $\Viol([n])_{\vec{x_i}}$. Since $\mathcal{A}$ is non-adaptive, we can assume it samples its randomness, makes all the queries and finally decides whether or not to reject the function.  So, let $r$ denote the randomness used by $\mathcal{A}$, and $ \mathcal{A}(r)$ be the set of queries made by the algorithm when the randomness is $r$. We have:
\begin{align*}
    \frac{2}{3} \leq \Pr_{r,i}[ \mathcal{A} \text{ rejects } \vec{x_{i}} \text{ when the randomness is } r]
    &\leq  \Pr_{r,i}[ \left(\mathcal{A}(r) \times \mathcal{A}(r)\right) \cap \Viol([n])_{\vec{x_{i}}} \neq \emptyset ] \\
     = \Pr_{r,i}[ \left(\mathcal{A}(r) \times \mathcal{A}(r)\right)  \cap \Viol(i)_{\vec{x_{i}}} \neq \emptyset ]
    &= \Pr_{r,i}[ \left(\mathcal{A}(r) \times \mathcal{A}(r)\right)  \cap \Twin(i) \neq \emptyset ]\,.
\end{align*}

In particular, there is a specific fixing of the randomness to $r^*$, such that:
\begin{align*}
    \frac{2}{3} &\leq \Pr_{i}[ \left(\mathcal{A}(r^*) \times \mathcal{A}(r^*)\right)  \cap \Twin(i) \neq \emptyset ] \,.
\end{align*}
    
That is, $\frac{2n}{3} \leq |\{i \;: \; \left(\mathcal{A}(r^*) \times \mathcal{A}(r^*)\right) \cap \Twin(i) \neq \emptyset\}|$. So by \Cref{lemma: set size to edges }, we must have $\frac{2}{3}n \leq |\mathcal{A}(r^*)| - 1$. That is, $\mathcal{A}$ must make at least $\frac{2n}{3}$ queries, as claimed.\Halmos
\end{proof}

\subsubsection{Two-sided adaptive lower bound\\}

We defer to \Cref{app:twosidedadaptive} the proof that, perhaps surprisingly, one cannot do better even with a two-sided adaptive algorithm. Specifically, the theorem whose proof we defer is the following:
\begin{theorem}\label{theorem: second lower bound}
Let $2^{-n}\leq \epsilon \leq 0.01$ and $\vec{x}:\{0,1\}^n \to \{0,1\}^n$. Let $\mathcal{A}$ be a two-sided adaptive tester that rejects $\vec{x}$ with probability at least $2/3$ when $d_\CMo(\vec{x}) \geq \epsilon$ and accepts $\vec{x}$ with probability at least $2/3$ when $d_\CMo(\vec{x})=0$. Then, $\mathcal{A}$ must have query complexity $\Tilde{\Omeg}(n/\epsilon)$.
\end{theorem}

\subsection{Testing coordinate-wise monotonicity (upper bounds)}
In this section, we consider the task of testing coordinate-wise monotonicity when bids can be integers between $0$ and $m$, and allocations can be between $0$ and $d$. Our goal is to prove the following theorem:
\begin{theorem}\label{theorem: grid tester}
 Given a function $\vec{x}:[0:m]^n \to [0:d]^n$ and $0<\epsilon\leq1$, there exists a one-sided non-adaptive tester with query complexity $\Omic(n \log(m)/\epsilon)$ to test if $\vec{x}$ is coordinate-wise monotone or $d_\CMo(\vec{x}) \geq \epsilon$.
\end{theorem}

The above result shows that one cannot hope to get a tester with sub-linear dependency on $n$. For coordinate-wise monotonicity the dimensions are independent, which isn't the case for monotonicity. Understanding the interactions between the dimensions of the hypercube played a key role in \citet{best_mono_tester}, who gave the state of the art monotonicity tester. Their algorithm is a one-sided non-adaptive tester for the monotonicity of functions $f:\{0,1\}^n \to \{0,1\}$ using $\Tilde{O}(\sqrt{n}/\epsilon^2)$ oracle queries. 

To this end, we first define "lines" for the hypergrid of our vector-valued functions.

\begin{definition}[{Lines of $\vec{x}:[0:m]^n \to [0:d]^n$}] 
    Let $d \geq 1$, and consider a function \\ $\vec{x}:[0:m]^n \to [0:d]^n$. Let $i \in [n]$, $\alpha \in [0:m]^{i-1},  \beta \in [0:m]^{n-i}$, the line $x_{\alpha,i, \beta}: [0:m] \to [0:d]$ is defined as \[x_{\alpha,i, \beta}(b) = \vec{x}(\alpha, b, \beta)_i \,.\]
    The set of $i$-lines of $\vec{x}$ is the set $\{ x_{\alpha,i, \beta} \; : \; \alpha \in [0:m]^{i-1},  \beta \in [0:m]^{n-i} \} $. 
\end{definition}

The below corollary follows from the definition of coordinate-wise monotonicity.
\begin{corollary}
Consider a function $\vec{x}: [0:m]^n \to [0:d]^n$. For a set $S \subseteq [n]$, we say that $x$ is $S$-monotone if for all $i \in S$, all the $i$-lines of $\vec{x}$ are monotone.  $\vec{x}$ is coordinate-wise monotone if and only if it is $[n]$-monotone.
\end{corollary}

We will use the following result from \citet{oded_improved} to define a distribution on pairs such that, if a function is far from monotone, a sampled pair detects a violation with sufficiently high probability.

\begin{lemma}[\cite{oded_improved}]\label{lemma: pair distribution }
    Let $m\geq1$, let $\Xi$ be a finite totally ordered set, and consider a function $f:[0:m]\to\Xi$. If $(p,l)$ is sampled according to $\calD$ of \Cref{eq:d}, then
    \[
        \Pr[f(p)>_\Xi f(l)]\geq \frac{d_\Mo(f)}{C\log(m)}
    \]
    for an absolute constant $C$. To define $\calD$, let $\nu_2(r)$ be the largest integer $t$ such that $2^t$ divides the positive integer $r$. Then $\calD$ is the uniform distribution over
    \begin{equation}
    \label{eq:d}
    \left\{(p,l)\in[0:m]^2: 0<l-p\leq 2^{\max\{\nu_2(p+1),\nu_2(l+1)\}}\right\}.
    \end{equation}
\end{lemma}

Based on this lemma, the tester we consider is the following one: \\
\begin{algorithm}[H]
\caption{One-sided non-adaptive coordinate-wise monotonicity tester for ${\Vec{x}:[0:m]^n \to [0:d]^n}$. }
\hspace*{\algorithmicindent} \textbf{Input: }{Distance parameter $0<\epsilon\leq1$ and query access to $\Vec{x} : [0:m]^n \to [0:d]^n$.}
\begin{algorithmic}[1]
 \For{$j = 1$ to $\left\lceil{{2 C n \log(m) } / {\epsilon}}\right\rceil$}
 \State Sample $i \in [n], \alpha \in [0:m]^{i-1}, \beta \in [0:m]^{n-i}$ uniformly at random. 
 \State Select pair $(p,l)$ according to the distribution of \Cref{eq:d}
 \If{$\Vec{x}(\alpha,p,\beta)_i > \Vec{x}(\alpha,l,\beta)_i $}\State Reject $\Vec{x}$. \EndIf
\EndFor
\State Accept $\Vec{x}$. 
\end{algorithmic}
\end{algorithm}

To show our tester works, we will need the following crucial lemma, the proof of which we defer to \Cref{app:proof lemma line dist}, and after which we can complete our proof.

\begin{lemma} \label{lemma: line dist}
    Let $\vec{x}:[0:m]^n \to [0:d]^n$ be $\epsilon$-far from being coordinate-wise monotone. If we sample a line $x_{\alpha, i, \beta}$ uniformly at random, its expected distance to monotonicity is at least $\eps/n$. That is, $\E_{\alpha, i, \beta} [d_\Mo( x_{\alpha, i, \beta})] \geq \epsilon / n$.
\end{lemma}

\begin{proof}{Proof of \Cref{theorem: grid tester}}
    The query complexity is $\Omic(n \log(m)/\epsilon)$, and a coordinate-wise monotone function is never rejected.

    Now let $\vec{x}$ be $\epsilon$-far from coordinate-wise monotonicity. By \Cref{lemma: pair distribution }, after sampling $i,\alpha,\beta$, the pair $(p,l)$ witnesses
    \[
        \Vec{x}(\alpha,p,\beta)_i=x_{\alpha,i,\beta}(p)>x_{\alpha,i,\beta}(l)=\Vec{x}(\alpha,l,\beta)_i
    \]
    with probability at least $d_\Mo(x_{\alpha,i,\beta})/(C\log(m))$. Since the line is uniform, \Cref{lemma: line dist} implies that the rejection probability in each iteration is at least
    \[
        \E_{\alpha,i,\beta}\left[\frac{d_\Mo(x_{\alpha,i,\beta})}{C\log(m)}\right]
        \geq \frac{\epsilon}{Cn\log(m)}.
    \]
    Thus the probability of failing to reject is at most
    \[
        \left(1-\frac{\epsilon}{Cn\log(m)}\right)^{2Cn\log(m)/\epsilon}
        \leq e^{-2}<\frac13.\Halmos
    \]
\end{proof}

\subsection{Testing mechanisms}
\label{subsec:main:full}

We now consider a richer model of functions, $\vec a: [0:m]^n \to [0:d]^n \times [0,dm]^n$. That is, we view the function $\vec{a}$ as a pair of functions $(\vec{x}, \vec{p})$, where $\vec{x}: [0:m]^n \to [0:d]^n$ describes the allocation to each bidder and $\vec p: [0:m]^n \to  [0,dm]^n$ describes the price they each pay. In particular, for $\Vec{b} \in [0:m]^n$, we denote by $\vec{a}(\Vec{b})_j$ the pair $(\vec{x}(\Vec{b})_j, \vec{p}(\Vec{b})_j)$.

As mentioned earlier, in this setting it might not make sense to consider the distance under the uniform distribution, but rather an arbitrary distribution $\mathcal{D}$ over $[0:m]^n$. In particular, it seems unlikely that each bidder has the same distribution of bids. Thus, we will consider that our algorithm has sample access to the distribution $\mathcal{D}$. We will consider a model where we can ask for conditional samples from $\mathcal{D}$, that is we can ask for a sample from $\mathcal{D}$ conditioned on the bids of $n-1$ bidders. %

We begin this section by the definitions needed to understand allocation functions and the correct price associated with them. 

\begin{definition}[Valid price function]
Let $x: [0:m] \to [0:d]$ be a monotone function. Let \[J_x \triangleq \{0 \} \cup \{ b : b \in [m] \text{ and } x(b) \neq x(b-1) \},\] $J_x$ records the bids for which we observe a "jump" in $x(b)$ and $0$.

Let $b \in [0:m]$, writing $J_x = \{t_0, t_1 \ldots, t_k\}$, where $t_i < t_{i+1}$, we denote by $i^b \in [0:k]$ the index of the unique $t_{i^b} \in J_x$ with $x(t_{i^b})=x(b)$. Note that in the above, we always have $t_0=0$.

Then, a valid price function for $x$ is any function $p : [0:m] \to [0,dm]$ such that there exists values $h_1, \ldots, h_k$ where:

\begin{itemize}
    \item For any $i\in [k]$, $t_i-1 \leq h_i \leq t_i$, and 
    \item for all $b \in [0:m]$, we have
\[
p(b) = 0 + \sum_{j=1}^{i^b} h_j \times (x(t_j) - x(t_{j-1}))\,.
\]
\end{itemize}

In particular, given $x$ a valid price function $p$ is uniquely defined by the values $h_i$, since given $x$ we can compute each $h_i$ as \[h_i= \left( p(t_i)- p(t_{i-1}) \right)  / \left(x(t_i) - x(t_{i-1})\right)\,.\] Furthermore, observe that we always have $p(0)=0$ and for any $b \in [0:m]$ we have $p(b)=p(t_{i^b})$, that is $p(b)$ is the same as the price at the smallest $b'$ with $x(b')=x(b)$. 

Given a function $\vec{x}: [0:m]^n \to [0:d]^n$, a valid price function for $\vec{x}$ is any function $ \Vec{p}: [0:m]^n \to [0,dm]^n$, such that for all $i \in [n], \alpha \in [0:m]^{i-1}, \beta \in [0:m]^{n-i}$, we have that
\[ \vec{p}(\alpha, 0, \beta)_i, \vec{p}(\alpha, 1, \beta)_i, \ldots, \vec{p}(\alpha, m, \beta)_i  \]
corresponds to a valid price function for $x_{\alpha, i, \beta}$.
\end{definition}

The above definition is the discrete DSIC characterization under the normalization $p_i(0,\vec b_{-i})=0$. We thus have the following definition of $\DSIC$.

\begin{definition}[$\DSIC$ function]\label{def: DSIC function}
    A function $\vec{a}=(\vec{x}, \vec{p}) : [0:m]^n \to [0:d]^n \times [0,dm]^n$ is said to be $\DSIC$ if and only if :
    \begin{itemize}
        \item $\vec{x}$ is coordinate wise monotone.
        \item $\vec{p}$ is a valid price function for $\vec{x}$.
    \end{itemize}
\end{definition}

Finally, the distance to being $\DSIC$ under a distribution $\mathcal{D}$ is given as follows:

\begin{definition}[Distance to $\DSIC$]
    Let $n \geq 1$. Let $\vec{a} =  (\vec{x},\vec{p}): [0:m]^n \to [0:d]^n \times [0,dm]^n$. Let $\mathcal{D}$ be a distribution over $[0:m]^n$, the distance of $\vec{a}$ to being $\DSIC$ under $\mathcal{D}$, denoted $d_{\DSIC}(\Vec{a})$ can be expressed as
 \[
 d^\mathcal{D}_{\DSIC}(\vec{a}) \triangleq \min_{\vec{y},\vec{q}} \Pr_{\Vec{b} \sim \mathcal{D}}\left[ \vec{x}(\vec{b} ) \neq \vec{y}( \vec{b}) \lor  \vec{p}(\vec{b} ) \neq \vec{q}( \vec{b}) \right]
 \,,
 \]
 where the minimum ranges over all coordinate-wise monotone functions $\vec{y}: [0:m]^n \to [0:d]^n$ and all valid price functions $\vec{q}: [0:m]^n \to [0,dm]^n $ for $\Vec{y}$.
\end{definition}

Note that under the above definition we can have $d^\mathcal{D}_{\DSIC}(\vec{a})=0$, but $\Vec{a}$ not being $\DSIC$. In such a case, the testers we present  might reject $\vec{a}$ even if its distance to $\DSIC$ is $0$ under $\mathcal{D}$. 

The main idea of our testers is to decompose $\vec{a}$ into its "lines", and consider each line as a separate function $a$ to test. 

\begin{definition}[{Lines of $\Vec{a}: [0:m]^n \to [0:d]^n \times [0,dm]^n$}]\label{def: lines with price}
    Consider a function $\vec{a}=(\vec{x}, \vec p): [0:m]^n \to [0:d]^n \times [0,dm]^n$. Let $i \in [n]$, $\alpha \in [0:m]^{i-1},  \beta \in [0:m]^{n-i}$. The line $a_{\alpha,i, \beta}: [0:m] \to [0:d] \times [0,dm]$ is described in two parts as $x_{\alpha,i, \beta} \text{ and }p_{\alpha,i, \beta}$.
    
    In particular for $b \in [0:m]$, we have 
    \[a_{\alpha,i, \beta}(b) = \left(x_{\alpha,i, \beta}(b), p_{\alpha,i, \beta}(b)\right) \,, \]
    where $x_{\alpha,i, \beta}(b)= \vec{x}(\alpha, b, \beta)_i $ and $p_{\alpha,i, \beta}(b)=  \vec{p}(\alpha,b, \beta)_i $. \\
    The set of $i$-lines of $\vec{a}$ is the set $\{ a_{\alpha,i, \beta} \; : \; \alpha \in [0:m]^{i-1},  \beta \in [0:m]^{n-i} \} $. 
\end{definition}

As mentioned before, we will assume that we have conditional sampling access to $\mathcal{D}$. This corresponds to fixing the bids of $n-1$ bidders, and sampling the bid of the remaining bidder. 
\begin{definition}[Line distribution of $\mathcal{D}$]
For $i\in[n]$, $\alpha\in[0:m]^{i-1}$, and $\beta\in[0:m]^{n-i}$, define
\[
\Pr_{L\sim\mathcal{D}_L}[L=(\alpha,i,\beta)]
=\frac1n\sum_{b=0}^m\Pr_{\vec B\sim\mathcal{D}}[\vec B=(\alpha,b,\beta)].
\]
\end{definition}

\begin{definition}[Projection of $\mathcal{D}$ to a line]
    Consider a distribution  $\mathcal{D}$ over $[0:m]^n$ given $i\in [n],\alpha \in [0:m]^{i-1},  \beta \in [0:m]^{n-i}$, we define its projection by $(\alpha, i, \beta)$ to be the distribution $\mathcal{D}_{\alpha,i, \beta}$ over $[0:m]$ defined as: 
    
    \[\Pr_{X \sim \mathcal{D}_{\alpha,i, \beta}}[ X = b]= \frac{ \Pr_{Y \sim \mathcal{D}}[ Y =  \alpha , b , \beta ] }{\sum_{b'=0}^m \Pr_{Y' \sim \mathcal{D}}[ Y' =  \alpha , b' , \beta]} \,. \]

    In the degenerate case where $\sum_{b'=0}^m \Pr_{Y' \sim \mathcal{D}}[ Y' =  \alpha , b' , \beta ]=0$, we simply define $\mathcal{D}_{\alpha,i, \beta}$ to be the uniform distribution over $[0:m]$. However, note that our algorithm will never ask for samples from $\mathcal{D}_{\alpha, i, \beta}$ in this case. 
\end{definition}

Hence, we will assume that given $(\alpha, i, \beta)$, our algorithm can get samples from $\mathcal{D}_{\alpha, i, \beta}$.

Using the above definitions, we defer the analysis of the tester to \Cref{app:earliertester} for brevity of exposition. The interested reader can refer to that section for the full details.
The broad idea is to use the following tester (\Cref{algo:mid_tester}) that we devise for a "line" of the hypergrid that operates using binary search to identify potential points of disagreement.  \\
\begin{algorithm}[H] 
\caption{One-sided adaptive $\DSIC$ tester for $a:[0:m] \to [0:d] \times [0,dm]$ under distribution $\mathcal{D}'$. }\label{algo:mid_tester}
\hspace*{\algorithmicindent} \textbf{Input: }{Distance parameter $0<\epsilon\leq1$, sample access to a distribution $\mathcal{D}'$ over $[0:m]$ and query access to $a=(x,p): [0:m] \to [0:d] \times [0,dm]$.}
 \begin{algorithmic}[1]
\State $S, H \gets \emptyset$. 
 \For{ $\ell=0$ to $d$}
 \State Perform binary search to try to find the smallest $b \in [0:m]$ with $x(b)=\ell$. 
 \If{such a value $b$ is found} 
 \State Record it as $v_\ell$ and add $\ell$ to $S$. 
 \EndIf
\EndFor
\State Write $S=\{s_0, \ldots, s_k\}$, such that $s_{i} < s_{i+1}$. And let $J_a \triangleq \{ t_i : t_i \triangleq v_{s_i}\}$. 
\If{$t_0 \neq 0$}
\State reject $a$.
\EndIf
\algstore{myalg}
 \end{algorithmic}
\end{algorithm}\clearpage

\begin{algorithm}[H]                     
\begin{algorithmic}[1] 
\algrestore{myalg}
 \For{ $i=1$ to $k$}
 \If{$t_{i-1} \geq t_{i}$} \State Reject $ a$. \EndIf 
 \State Let $h_i=\left(p(t_i)-p(t_{i-1})\right) / \left(x(t_{i})-x(t_{i-1})\right)$ and add $h_i$ to $H$. 
 \If{$t_{i} < h_i \text{ or } h_i < t_i-1 $}
 \State Reject $ a$.
 \EndIf
\EndFor
\State Consider the (monotone) function $y$ defined as:
\[
y(b) = x(t^b) \text{ where $t^b$ is the largest $t \in J_a$ with $t^b \leq b$.}
\]
\State For any $b \in [0:m]$, denote by $i^b$ the index of the unique $t \in J_a$ such that $x(t_{i^b})=y(b)$. Consider the valid price function $p^y$ for $y$ defined as:
\[
p^y(b) = 0 + \sum_{j=1}^{i^b} h_j \times (x(t_j) - x(t_{j-1})).
\]
\For{ $j=1$ to $\lceil2/\epsilon\rceil$ }
Sample a point $b$ from $\mathcal{D}'$ and query $a(b)$. 
    \If{$(x(b),p(b)) \neq (y(b), p^y(b))$}
    \State Reject $a$.
    \EndIf
\EndFor
\State Accept $a$. 
 \end{algorithmic}
\end{algorithm}

Then, the formal theorem that we will prove in \Cref{app:earliertester} follows:
\begin{theorem}
\label{thm:ub}
Given an allocation function $\vec{a}:[0:m]^n \to [0:d]^n \times [0,dm]^n$, conditional sampling access to a distribution $\mathcal{D}$ over $[0:m]^n$, and $0<\epsilon\leq1$, there exists a one-sided tester with query complexity
\[
\Omic\left(
\frac{n}{\epsilon}\log\left(\frac{2n}{\epsilon}\right)
\min\{m,d\log(m)\}
+\frac{n}{\epsilon}\log^2\left(\frac{2n}{\epsilon}\right)
\right)
\]
to test if $\vec{a}$ is $\DSIC$ or $d^\mathcal{D}_\DSIC(\vec a) \geq \epsilon$. The tester uses $\Omic((n/\epsilon)\log^2(2n/\epsilon))$ conditional samples and runs in time $\Tilde{\Omic}(n^2d/\epsilon)$ under exact unit-cost arithmetic, charging $\Tilde{\Omic}(n)$ time per query or sample.
\end{theorem}

The tester we consider to prove the above theorem is the following one: \\
\begin{algorithm}[H] 
\caption{One-sided $\DSIC$ tester for ${\Vec{a}:[0:m]^n \to [0:d]^n \times [0,dm]^n}$ under distribution $\mathcal{D}$.}
\label{algo:newpricetester}
 \hspace*{\algorithmicindent} \textbf{Input: }{Distance parameter $0<\epsilon\leq1$, query access to $\Vec{a} : [0:m]^n \to [0:d]^n \times [0,dm]^n$ and conditional sampling access to $\mathcal{D}$. }
\begin{algorithmic}[1]
\State Let $T$ be  \Cref{algo:mid_tester}. Let $\epsilon'=\epsilon/2n$. 
\State Let $\ell= \left\lceil \log_2 \frac{2}{\epsilon'} \right\rceil$.
 \For{$j=1$ to $\ell$}
\State Repeat $\left\lceil 12\ell/(2^j\epsilon')\right\rceil$ times:
\Indent
 \State Sample $\Vec{b} \sim \mathcal{D}$. Sample $i \in [n]$ uniformly at random. 
\State  Fix $\alpha=\Vec{b}_1, \ldots, \vec{b}_{i-1}$ and $\beta=\Vec{b}_{i+1}, \ldots, \vec{b}_{n}$. 
\State Run $T$ twice independently on $a_{\alpha, i, \beta}$, with sample access to the distribution $\mathcal{D}_{\alpha,i,\beta}$ and distance parameter \hspace*{\algorithmicindent}\hspace*{\algorithmicindent}$2^{-j}$.
\If{$T$ rejected $a_{\alpha, i, \beta}$ at least once}\State Reject $\Vec{a}$.
\EndIf
\EndIndent
\EndFor
\State Accept $\vec{a}$.
\end{algorithmic}
\end{algorithm}

\subsection{A tester with better query complexity}
In this section, we present the key ideas of our $\DSIC$ tester for functions $a : [0:m] \to [0:d] \times [0:dm]$ with better query complexity than \Cref{algo:mid_tester} whose query complexity and running time are given in \Cref{lemma: line DSIC tester richer alloc}. The query complexity of the tester is $\Omic(\log(m)/\epsilon)$. However, this tester has a worse running time than the tester presented before, in particular it runs in time $d^{\Omic(m)}$ up to polynomial factors. Using this tester, we can achieve a $\DSIC$ tester with query complexity $\Tilde{\Omic}(n/\epsilon)$ for functions $\vec{a} : [0:m]^n \to [0:d]^n \times [0:dm]^n$.

More specifically, given $a : [0:m] \to [0:d] \times [0:dm]$, we can assume $m+1$ is a power of $2$. Indeed, if $m+1$ isn't a power of two, we can proceed as follows, let $
    \overline{m}+1
    :=
    2^{\left\lceil \log_2(m+1)\right\rceil}$.
Then
    $m \leq \overline{m} < 2m+1.$
So, we extend $a \colon [0:m] \to [0:d]\times[0:dm]$
to a function
\[
    \overline{a}\colon [0:\overline{m}]
    \to [0:d]\times[0:d\overline{m}]
\]
by setting
\[
    \overline{a}(b)
    :=
    \begin{cases}
        a(b), & 0\leq b\leq m,\\
        a(m), & m<b\leq \overline{m}.
    \end{cases}
\]
Likewise, we extend the distribution $\mathcal D$ over $[0:m]$ to a
distribution $\overline{\mathcal D}$ over $[0:\overline{m}]$ by setting
\[
    \overline{\mathcal D}(b)
    :=
    \begin{cases}
        \mathcal D(b), & 0\leq b\leq m,\\
        0,             & m<b\leq \overline{m}.
    \end{cases}
\]

It's easy to see that $d_{\mathrm{DSIC}}^{\overline{\mathcal D}}(\overline{a})
    =
    d_{\mathrm{DSIC}}^{\mathcal D}(a)$ and we can easily simulate query access to $\overline{a}$ by querying $a$.  Hence, from now on, we assume $m+1=2^k$ for some $k \geq 1$ and consider a binary tree $T_a$ of depth $k$.

Each node in the tree corresponds to an interval $[s:t]$, where the leafs of the tree correspond to the intervals $[0:0], \ldots, [m:m]$. For simplicity we might sometimes refer to a leaf labeled with $[b:b]$ as simply being labeled with $b$. 

For each node $u$ we keep track of a set $P_u$ of $\DSIC$ function. Given a node $u$ with parent $v$, the set $P_u$ depends on the sets $P_v$ as well as a constant number of extra queries made to $a$. Hence for any leaf $[b:b]$, one can compute $P_{[b:b]}$ using $\Omic(k)=\Omic(\log(m))$ queries.

In particular, the tree $T_a$ is formally defined as follow: 

\begin{definition}[Tree $T_a$]\label{def: tree T_a}
Assume $a:[0:m] \to [0:d] \times [0,dm]$, where $m+1=2^k$. $T_a$ is a binary decision tree of height $k+1$.

The root of $T_a$ is labeled with $[0:m]$. Then if a node is labeled with an interval $[s:t]$, its left child is labeled with $[s:\frac{s+t-1}{2}]$ (the first half of the interval) and the right child with with $[\frac{s+t+1}{2}:t]$ (the second half). 

 For each internal node $u$ we will consider two sets of functions $P_u, T_u$ these are sets of $\DSIC$ functions. The leaves of the tree are only labeled with an interval $P_u$. The intervals are defined in a top-down fashion as follow:
\begin{itemize}
    \item The root is labeled with \[\mathcal{P}_{[0:m]} = \{ f \;|\; f \text{ is $\DSIC$, } f(0)=a(0) \text{ and } f(m)=a(m) \}\,.\] This is simply the set of $\DSIC$ function consistent with $a$ on $0$ and $m$.
    \item Then for an internal node $u=[s:t]$, we define \[\mathcal{T}_{u} = \{ f \;|\; f \in \mathcal{P}_{u},\, f\left(\frac{s+t-1}{2}\right)=a\left(\frac{s+t-1}{2}\right) \text{ and } f\left(\frac{s+t+1}{2}\right)=a\left(\frac{s+t+1}{2}\right) \} \,.\] That is, we restrict $\mathcal{P}_{u}$ to the functions which agree with $a$ on the two middle points of $u$.

 \item For an non-root node $u=[s:t]$ with parent $p$, we set $\mathcal{P}_u=\mathcal{T}_p$.
\end{itemize}
\end{definition}

The proof of correctness of this second tester is more challenging than the one of \Cref{algo:mid_tester}. The key properties of $T_a$ are captured in the following two theorems. They relate the distance of $a$ to being $\DSIC$ under a distribution $\mathcal{D}$ to the probability that a leaf $[b:b]$ of the tree, where $b \sim \mathcal{D}$, is such that $\mathcal{P}_{[b:b]}=\emptyset$.

\begin{theorem}\label{theorem:tree_DSIC}
Consider a function $a : [0:m] \to [0:d] \times [0,dm]$ such that $a$ is $\DSIC$. Then, for any node $u$ of $T_a$, we have $\mathcal{P}_u \neq \emptyset$. Furthermore, if $u$ is an internal node, we also have $\mathcal{T}_u \neq \emptyset$.
\end{theorem}
\begin{theorem}\label{theorem:tree_not_DSIC}
Let $\mathcal{D}$ be a distribution and consider a function $a : [0:m] \to [0:d] \times [0,dm]$. 
If $d_{\DSIC}^\mathcal{D}(a) \geq \epsilon$, then in the tree $T_a$ we have \[\Pr_{b \sim \mathcal{D}}\left[ \mathcal{P}_{[b:b]}=\emptyset \right] \geq \epsilon \,. \]
\end{theorem}

The proof of the above two theorems can be found in \Cref{app:tree_lemmas}. As such, our tester for functions $a : [0:m] \to [0:d] \times [0:dm]$ simply samples $O(1/\epsilon)$ many $b \sim \mathcal{D}$ and rejects $a$ if for any of the sampled points $\mathcal{P}_{[b:b]}=\emptyset$. Note that checking if $\mathcal{P}_{[b:b]}=\emptyset$ is a non trivial task, in particular to solve we problem, we solve $(d+1)^{m+1}$ many linear programs. This explains the running time of our algorithm.

Using this tester as a subroutine, we can obtain the following theorem :

\begin{theorem}
\label{thm:ub_amazing}
Given an allocation function $\vec{a}:[0:m]^n \to [0:d]^n \times [0,dm]^n$, conditional sampling access to a distribution $\mathcal{D}$ over $[0:m]^n$, and $0<\epsilon\leq1$, there exists a one-sided tester with query complexity
\[
\Omic\left(\frac{n}{\epsilon}\log(m)\log^2\left(\frac{2n}{\epsilon}\right)\right)
\]
to test if $\vec{a}$ is $\DSIC$ or $d^\mathcal{D}_\DSIC(\vec a) \geq \epsilon$. The tester uses $\Omic((n/\epsilon)\log^2(2n/\epsilon))$ conditional samples and runs in time
\[
\Tilde{\Omic}\left(\frac{n^2}{\epsilon}\right)
+\Tilde{\Omic}\left(\frac{n}{\epsilon}\right)d^{\Omic(m)}
\]
under the unit-cost convention of \Cref{thm:ub}.

\end{theorem}

The proof of the above theorem, as well the formal description and analysis of our second $\DSIC$ tester for functions $a :[0:m] \to [0:d]\times[0:dm]$ can be found in \Cref{app:proof of amazing theorem}.

\section{Conclusion and Future Directions}

This paper opens up a new area, the intersection between two fields, those of algorithmic game theory and property testing.
In this work, we investigated the hardness of testing whether a mechanism $(\Vec{x}, \Vec{p})$ is IC or $\epsilon$-far from any IC mechanism and gave a one-sided adaptive tester for this task. We also showed that testing whether a function $\Vec{x}$ is coordinate-wise monotone or $\epsilon$-far from it requires $\wt{\Omeg}(n/\epsilon)$ queries, even for a two-sided adaptive tester, and gave an algorithm with a matching $\wt{\Omic}(n/\epsilon)$ upper bound.

We now present a few open questions arising from this new area that we believe would be well-motivated for future research:
\begin{itemize}
    \item Recall that one of our one-sided $\DSIC$ testers for functions $(\Vec{x}, \Vec{p}) : [0:m]^n \to [0:d]^n \times [0,dm]^n$ with query complexity $\Tilde{\Omic}(n/\epsilon)$ is optimal (up to log factors) when it comes to query complexity, but it has exponential running time. We also gave a tester with polynomial running time but query complexity
    \[
    O\left(\frac{n \log(2n/\epsilon)}{\epsilon}
    \left(\log(2n/\epsilon)+\min\{m,d\log(m)\}\right)\right).
    \]
    Can this computational gap be closed, or is there a provable separation? If the latter is true, this would enter into the (few) separations known in computational learning theory.
    \item Our second one-sided $\DSIC$ tester had query complexity $\Tilde{\Omic}\left( \frac{n \log(m)}{\epsilon} \right)$. Can we also get a tester with complexity $\Tilde{\Omic}\left( \frac{n \log(d)}{\epsilon} \right)$? Such an improvement would be well-motivated in a setting with a lot of players (dependent on dimension $m$) and few items (dependent on dimension $d$). We conjecture this to be true: \citet{10.1145/3155296} gave a one-sided tester for the monotonicity of functions $f:[0:m]^n \to [0:d]$ with $\Tilde{\Omic}\left(\frac{n\log(d)}{\epsilon}\right)$ queries. Hence, we could hope that such an approach leads to a tester with no dependency on $m$. However, \citet{10.1145/3155296} crucially rely on the definition of distance under the uniform distribution to prove the correctness of their algorithm, whereas we would like it with respect to an unknown distribution $\calD$.

\end{itemize}

\section*{Acknowledgments}
This research was supported in part by NSF awards CNS-2212745, CCF-2212233, DMS-2134059, and CCF-1763970. Jason Milionis' research was also supported by an Onassis Foundation Scholarship and an A.G. Leventis educational grant.

\printbibliography

\appendix
\section{Deferred Proofs}\label{sec:appendix}

\subsection{Proof of \Cref{theorem: second lower bound}}
\label{app:twosidedadaptive}

We first prove the following lemma for our analysis.
\begin{lemma}\label{lemma: points to edges}
     Given a set $Q \subseteq \{0,1\}^n$, we have 
    \begin{align*} | \{ (\vec{b},\vec{b'}) \in Q \times Q \;: \; (\vec{b},\vec{b'}) \in \Twin([n]) \}| \leq |Q|\log_2|Q|.\end{align*}

    In particular, given a set of $Q$ of $q$ points from the hypercube, there is at most $q \log_2(q)$ different edges such that both end points of the edge are in $Q$. 
\end{lemma} 

\begin{proof}{Proof of \Cref{lemma: points to edges}}
Let $n \geq 1$ be the dimension of the hypercube, and $Q$ be a set of points of the $n$ dimensional hypercube with $|Q| \geq 1$. Let $\vec{q} \in \{0,1\}^{2^n}$ be the indicator vector of what point of the hypercube are in $Q$. To prove the lemma, we need to bound $\frac{1}{2}(\vec{q}^t Q_n \vec{q})$, where $Q_n$ is the adjacency matrix of the $n$ dimensional hypercube. $Q_n$ has a recursive structure, in particular, letting $I_n$ be the $2^n \times 2^n$ identity matrix:
\begin{align*} Q_n = \begin{bmatrix} Q_{n-1} \; I_{n-1} \\ I_{n-1} \; Q_{n-1} \end{bmatrix} \text{ and  } Q_1 = \begin{bmatrix} 0 \quad 1 \\ 1 \quad 0 \end{bmatrix} \end{align*}

The proof is by induction on $n$ and our goal is to show that 
\begin{align*}(\vec{q}^t Q_n \vec{q}) \leq \| \vec{q} \|_1 \log_2( \|\vec{q} \|_1)\,.\end{align*}

The case where $n=1$ is easy. So let $n \geq 2 $, and decompose $\vec{q}=[\vec{q_1}; \vec{q_2}]$, where $\vec{q_i} \in \{0,1\}^{2^{n-1}}$ and note that one of $\vec{q_1}, \vec{q_2}$ isn't the all $0$ vector. 

\begin{align*}
    \vec{q}^t Q_n \vec{q} &= \begin{bmatrix}
        \vec{q_1}^t \; \vec{q_2}^t
    \end{bmatrix} \begin{bmatrix} Q_{n-1} \; I_{n-1} \\ I_{n-1} \; Q_{n-1} \end{bmatrix} \begin{bmatrix}
        \vec{q_1} \\ \vec{q_2}
    \end{bmatrix} \\
    &= \vec{q_1}^T Q_{n-1} \vec{q_1} + \vec{q_2}^T Q_{n-1} \vec{q_2} + \vec{q_1}^T \vec{q_2} + \vec{q_1} \vec{q_2}^T \\
    & \leq  \vec{q_1}^T Q_{n-1} \vec{q_1} + \vec{q_2}^T Q_{n-1} \vec{q_2} + 2\min( \|\vec{q_1}\|_1, \|\vec{q_2}\|_1)
\end{align*}
If $\vec{q}_2$ is the all $0$ vector by the inductive hypothesis, we get:
\begin{align*}\vec{q}^t Q_n \vec{q}  = \vec{q_1}^T Q_{n-1} \vec{q_1} \leq  \|\vec{q}\|_1 \log(\|\vec{q}\|_1).\end{align*}
We have the same result if $\vec{q}_1$ is the all $0$ vector.

Otherwise,  we can apply the inductive hypothesis to both $\vec{q_1}^T Q_{n-1} \vec{q_1}$ and $\vec{q_2}^T Q_{n-1} \vec{q_2}$, to obtain
\begin{align*}
    \vec{q}^t Q_n \vec{q} &\leq  \|\vec{q_1}\|_1 \log(\|\vec{q_1}\|_1) + \|\vec{q_2}\|_1 \log(\|\vec{q_2}\|_1) + 2\min( \|\vec{q_1}\|_1, \|\vec{q_2}\|_1)
\end{align*}

Let $k = \|\vec q\|_1$, $x = \|\vec{q_1}\|_1$, and assume wlog $\|\vec{q_1}\|_1 \leq \|\vec{q_2}\|_1$, so that we have $1 \leq x \leq k/2$. The expression becomes 
\begin{align*}
\vec{q}^t Q_n \vec{q} \leq x \log(x)+(k-x)\log(k-x)+2x
\,.
\end{align*}
Maximizing the right hand side yields a function value of $k\log k$ by setting $x=k/2$ (the right hand side is convex in the domain of $1 \leq x \leq k/2$, and we evaluate local extrema and endpoints).
Hence, we get
\begin{align*}
\vec{q}^t Q_n \vec{q} \leq k \log k = \|\vec{q}\|_1 \log\|\vec{q}\|_1
\,,
\end{align*}
as claimed.\Halmos
\end{proof}

Before proving \Cref{theorem: second lower bound}, we first define different distributions of functions:
\begin{itemize}
    \item Distribution $D_1$, for each $i \in [n]$ let $(\vec{b}, \vec{b'}) \in \Twin(i)$. We set $\vec{x}(\vec{b})_i, \vec{x}(\vec{b'})_i$ such that $\vec{x}(\vec{b})_i=\vec{x}(\vec{b'})_i=1$ with probability $1/2$ and $\vec{x}(\vec{b})_i=\vec{x}(\vec{b'})_i=0$ otherwise.
    \item Let $j \in [n]$, we define the distribution $D_{0,j}$ as follows:
    \begin{itemize} 
        \item For each edge $(\vec{b}, \vec{b'}) \in \Twin(j)$, we sample $(\vec{b}, \vec{b'})$ with probability $80\epsilon$. Let $S$ be the set of all sampled edges.
        For each $(\vec{b}, \vec{b'}) \in \Twin(j)$ if the edge was sampled, we set both $\vec{x}(\vec{b})_j$ and $\vec{x}(\vec{b'})_j$ independently at random from $\{0,1\}$. Else, we set $\vec{x}(\vec{b})_j, \vec{x}(\vec{b'})_j$ as $\vec{x}(\vec{b})_j=\vec{x}(\vec{b'})_j=1$ with probability $1/2$ and otherwise $\vec{x}(\vec{b})_j=\vec{x}(\vec{b'})_j=0$ otherwise.

        \item For each $i \in [n] \setminus \{j\}$, let $(\vec{b}, \vec{b'}) \in \Twin(i)$. We proceed as in $D_1$ and set $\vec{x}(\vec{b})_i, \vec{x}(\vec{b'})_i$ as $\vec{x}(\vec{b})_i=\vec{x}(\vec{b'})_i=1$ with probability $1/2$ and $\vec{x}(\vec{b})_i=\vec{x}(\vec{b'})_i=0$ otherwise.
    \end{itemize}
\end{itemize}

Finally we define the distribution $D_0$ by first sampling $j \in [n]$ uniformly at random and then sampling a function from $D_{0,j}$.
\begin{lemma}
    Let $2^{-n}\leq \epsilon \leq 0.01$. Any function sampled from $D_1$ is coordinate-wise monotone, and with probability at least $19/20$ a function sampled from $D_0$ is $\epsilon$-far from being coordinate-wise monotone.
\end{lemma}
\begin{proof}{Proof}
The first part of the lemma is trivial. So, consider a function $\vec{x}$ sampled from $D_{0,j}$. Notice that for $i\neq j$, we have $\Viol(i)_{\Vec x}=\emptyset$. Consider an edge $(\vec{b}, \Vec{b'}) \in \Twin(j)$. If the edge wasn't sampled, then there is no need to change $\vec{x}(\vec{b})_j$ and $\vec{x}(\Vec{b'})_j$ as there is no violation along that edge. Otherwise, $\vec{x}(\vec{b})_j$ and $\vec{x}(\vec{b'})_j$ are set independently at random, so with probability $\frac{1}{4}$, we have $\vec{x}(\vec{b})_j=1$ and $\vec{x}(\vec{b'})_j=0$, making $(\vec{b}, \Vec{b'})$ a violating edge. It's clear that for each edge in $|\Viol(j)|$, we must change the allocation at one of the endpoints of the function to make $\vec{x}$ coordinate-wise monotone. Thus, $d_\CMo(\vec{x}) \geq \frac{|\Viol(j)|}{2^n}$.

For each edge $e \in \Twin(j)$ let $t_e$ be the indicator variable that records whether $e$ notices a violation or not. We have $\E[t_e]=   \frac{80\epsilon}{4}$. Since all edges are independent, by Chernoff bound, 

\begin{align*} \Pr[\sum_{e} t_e \leq \frac{ 80\epsilon }{4} \times 2^{n-1} \times (1 - \frac{4}{5} ) ] \leq e^{  -10 \epsilon 2^{n-1} (4/5)^2 } \leq \frac{1}{20}.\end{align*}
Since $\epsilon 2^{n-1} \geq 1/2$ \footnote{We can always assume $\epsilon \geq 2^{-n}$, since it's the minimum distance between two non equal functions with domain $\{0,1\}^n$}.
Thus, with probability at most $1/20$, less than $80\epsilon \times \frac{2^{n-3}}{5}$ edges in the $j$-th direction are violating in $\vec{x}$. So, with probability at least $19/20$, a function sampled from $D_{0,j}$ is $\epsilon$ far from being coordinate-wise monotone.\Halmos
\end{proof}

\begin{proof}{Proof of \Cref{theorem: second lower bound}}
    Let $2^{-n}\leq \epsilon \leq 0.01$. And let $\mathcal{A}$ be a randomized adaptive, two-sided, algorithm for testing coordinate-wise monotonocity with query complexity $q$. Wlog, we can assume that $\mathcal{A}$ always makes exactly $q$ queries and never queries the same point twice. In our setting, when $\mathcal{A}$ queries a point $\vec{b}$, we will answer with $\vec{x} (\vec{b})$ and sometimes when $d_\CMo(\vec{x}) \geq \epsilon$ we can answer with $\sucx$ (so the tester knows the function isn't coordinate-wise monotone and can immediately reject). This is slight generalization of the usual setting, but clearly a lower bound in this model is a lower bound for the usual setting where a query is always replied to with $\vec{x}(\vec{b})$.

    We can view $\mathcal{A}$ as a randomized mapping from query/answer history $[(\vec{b_1},a_l), \ldots,  (\vec{b_l},a_l)]$ to $\vec{b_{l+1}}$ for $l<q$ and to $\{\mathsf{Accept},  \mathsf{Reject}\}$ when $q=l$. Here, $a_l$ can be either one of $\vec{x}(\vec{b_l})$ or $\sucx$. 

    We will describe two processes $P_0$, $P_1$ that correspond to sampling a function from $D_0$ and $D_1$ respectively, as they answer queries from $\mathcal{A}$. Given $k \in \{0,1\}$, $H_k$ is the distribution of the final query-answer history of $\mathcal A$ when interacting with $P_k$. That is $H_k$ is the query/answer $[(\vec{b_1},a_1), \ldots,  (\vec{b_q},a_q)] \in 
    \left(\;\{0,1\}^n 
    \times \left( \{0,1\}^n \cup \{\sucx \} \right) \; \right)^q$ produced once $\mathcal{A}$ has finished its queries. Thus, the probability $\mathcal{A}$ accepts a function from $D_k$ is 
    \begin{align*}
        \sum_{(\vec{b_1},a_1), \ldots,  (\vec{b_q},a_q)}\Pr[H_k=[(\vec{b_1},a_1), \ldots,  (\vec{b_q},a_q)] ] \times \Pr[ \mathcal{A}\left((\vec{b_1},a_1), \ldots,  (\vec{b_q},a_q)\right)=\mathsf{Accept}]
    \end{align*}
    
Note that $\mathcal{A}$ must accept a function from $D_1$ with probability at least $2/3$ since all function from $D_1$ are coordinate-wise monotone. However, with probability at least $19/20$ a function from $D_0$ is $\epsilon$ far from coordinate-wise monotone and such functions must be rejected by the algorithm with probability at least $2/3$. As such, $\mathcal{A}$ can accept function sampled from $D_0$ with probability at most $11/30$.  Thus, the statistical difference between $H_0$ and $H_1$ must be at least $2/3-11/30= 0.3$. We will show that if $q < \frac{n}{5600\epsilon \log(n/\epsilon)}$, then the statistical difference between $H_0$ and $H_1$ is stricly less than $0.3$. \\
    
    The process $P_1$ proceeds as follows when $\mathcal{A}$ queries $\vec{b}$, $P_1$ sets $\vec{x}(\vec{b})$ bit by bit. For each $i \in [n]$: \begin{itemize}
        \item If $\mathcal{A}$ has already queried the $i$-twin of $\vec{b}$, set $\vec{x}(\vec{b})_i=\vec{x}(\vec{b^{(i)}})_i$.
        \item Else set $\vec{x}(\vec{b})_i$ as $0$ or $1$ uniformly at random.
    \end{itemize}
    Finally, $P_1$ answers with $\vec{x}(\vec{b})$. After $\mathcal{A}$ made all its queries, $P_1$ samples a function from $D_1$ uniformly at random that is consistent with the function $\vec{x}$ it built so far. It's easy to see that $P_1$ produces a function according to the distribution $D_1$.

    The process $P_0$ proceeds as follow. First, the algorithm samples $j \in [n]$ uniformly at random. Then, for each edge $(\vec{b}, \vec{b'}) \in \Twin(j)$, we sample $(\vec{b}, \vec{b'})$ with probability $80\epsilon$. Let $S$ be the set of sampled edges.  When $\mathcal{A}$ queries $\vec{b}$, $P_0$ sets $\vec{x}(\vec{b})$ bit by bit. 
    \begin{itemize}
        \item For each $i \in [n] \setminus \{j\}$ 
        \begin{itemize}
        \item If $\mathcal{A}$ has already queried the $i$-twin of $\vec{b}$, set $\vec{x}(\vec{b})_i=\vec{x}(\vec{b^{(i)}})_i$.
        \item Else set $\vec{x}(\vec{b})_i$ as $0$ or $1$ uniformly at random.
    \end{itemize}
           \item If $(\vec{b}, \vec{b^{(j)}})$ (or $(\vec{b^{(j)}}, \vec{b})$) was sampled in $S$, then  $P_0$ sets $\vec{x}(\vec{b})_j$ as $0$ or $1$ uniformly at random. 
           \item Else:
           \begin{itemize}
               \item If $\vec{b^{(j)}}$ wasn't already queried by $\mathcal{A}$, then $P_0$ sets $\vec{x}(\vec{b})_j$ as $0$ or $1$ uniformly at random.
           \item Else, if $\mathcal{A}$ has already queried the $j$-twin of $\vec{b}$, set $\vec{x}(\vec{b})_j=\vec{x}(\vec{b^{(j)}})_j$.
           \end{itemize}
           
    \end{itemize}
    Finally, if $\mathcal{A}$ had already queried $\vec{b^{(j)}}$ and $(\vec{b}, \vec{b^{(j)}}) \in S$ (or $(\vec{b^{(j)}}, \vec{b})  \in S$), $P_0$ answers the query with $\sucx$. Else, it replies with $\vec{x}(\vec{b})$.
    
    After $\mathcal{A}$ made all its queries, $P_0$ samples a function from $D_{j,0}$ uniformly at random that is consistent with the function $\vec{x}$ it built so far. It's easy to see that $P_0$ produces a function according to the distribution $D_0$.

    When $\mathcal{A}$ interacts with $P_1$ and makes its $l$-th query, then for each $i \in [n]$, the bit $\vec{x}(\vec{b_l})_i$ is sampled uniformly at random, or is set to $\vec{x}(\vec{b_l^{(i)}})_i$ if the $i$-twin of $\vec{b_l}$ was already queried. When $\mathcal{A}$ interacts with $P_0$, and makes its $l$-th query, if $\vec{b_l^{(j)} }$ was already queried, where $j$ was the coordinate sampled at first by the process and the edge $(\vec{b_l}, \vec{b_l^{(j)} })$ was sampled, then the query is replied to with $\sucx$.  If this doesn't happen, then for each $i \in [n]$, the bit of $\vec{x}(\vec{b_l^{(i)} })_i$ is set according to the same process as in $P_1$. In particular, if $P_0$ doesn't reply with $\sucx$, the answer of $P_0$ and $P_1$ are indistinguishable.

    Thus, conditioned on the event that $P_0$ never replied $\sucx$ to any of the queries, the distribution of query-answer histories $H_0$ and $H_1$ are identical. So, an upper bound on the statistical distance between $H_1$ and $H_0$ is the probability that that $P_0$ replies $\sucx$ to some query of $\mathcal{A}$. 

    By \Cref{lemma: points to edges}, if $\mathcal{A}$ perform $q$ queries, there are at most $q \log_2(q)$ different edges such that $\mathcal{A}$ queried both end points of the edge. $P_0$ initially samples $j \in [n]$, and a set $S$ of edges in $\Twin(j)$ and replies $\sucx$ only if $\mathcal{A}$ queries $\vec{b}$ and has already queried its $j$-twin where $(\vec{b}, \vec{b^{(j)}})$ (or $(\vec{b^{(j)}}, \vec{b})$) was sampled. 

    Assume for contradiction $q \leq \frac{n}{5600\epsilon \log(n/\epsilon)}$. Then, $\mathcal{A}$ queries the two end points of at most $\frac{n}{5600 \epsilon}$ edges of the hypercube. By a Markov bound, we get that with probability at least $6/7$, $\mathcal{A}$ has queried the two end points $\vec{b},\vec{b^{(j)}}$ of less than $\frac{1}{800 \epsilon}$ many edges in the $j$-th direction. 

    Now each edge in $\Twin(j)$ is independently sampled with probability $80\epsilon$. So conditioned on $\mathcal{A}$ querying the two end points of at most  $\frac{1}{800 \epsilon}$ many edges in the $j$-th direction, the probability that the algorithm queries the two end points of some edge in $S$ is at most
    $ 1- (1- 80\epsilon)^{1/800\epsilon} \leq 0.183 $
    whenever $\epsilon \leq 0.01$.
    Hence, with probability at least $0.817\times \frac{6}{7} > 0.7$, $P_0$ doesn't answer $\sucx$. Thus, if $q \leq \frac{n}{5600\epsilon \log(n/\epsilon)}$, we have that the statistical difference between $H_1$ and $H_0$ is $<0.3$, contradicting the correctness of $\mathcal{A}$. Therefore, we must have $q = \wt{\Omeg}( \frac{n}{\epsilon})$.\Halmos
\end{proof}

\subsection{Proof of \Cref{lemma: line dist}}
\label{app:proof lemma line dist}

We first define the line sorting operator.
\begin{definition}[Line sorting operator]\label{def : lso}
    Consider a function $\vec{x}:[0:m]^n \to [0:d]^n$. The function $L^i[\vec{x}]: [0:m]^n \to [0:d]^n$ is defined as follows: 
    For every $\alpha \in [0:m]^{i-1},  \beta \in [0:m]^{n-i}$, we let $L^i[\vec{x}](\alpha, 0, \beta)_i , \ldots , L^i[\vec{x}](\alpha, m, \beta)_i$  be $y(0) \ldots, y(m)$ where $y: [0:m] \to [0:d]$ is the monotone function closest to $x_{\alpha, i, \beta}$. 
    Finally, for $j \neq i$, $\Vec{b} \in [0:m]^n$, we always have $L^i[\vec{x}](\Vec{b})_j=\vec{x}(\Vec{b})_j$.
\end{definition}
In particular, the line sorting operator $L^i$ makes all the $i$-lines of $\vec{x}$ monotone by mapping them to the closest monotone function, and it does so by only changing the $i$-th bit of $\vec{x}$.

We now show some crucial properties of our operator:
\begin{lemma}
    Consider a function $\vec{x}:[0:m]^n \to [0:d]^n$, $\vec{b} \in [0:m]^n$. Let $\vec{x_0}=\vec{x}$ and $\vec{x_i}= L^i[\ldots L^2[L^1[\vec{x}]]\ldots]$. We have the following:
    \begin{itemize}
        \item For  $n \geq k > i \geq 0$ we have $\vec{x_i}(\vec{b})_k = \vec{x}(\vec{b})_k$.
        \item $\vec{x_n}$ is coordinate-wise monotone.
        \item $L^{i+1}[\vec{x_{i}}](\vec{b})_{i+1} = L^{i+1}[\vec{x}](\vec{b})_{i+1}$.
    \end{itemize}
\end{lemma}

\begin{proof}{Proof}
    $L^{i}[\vec{x}]$ makes all the $i$-lines of $\vec{x}$ monotone, and in particular $L^i$ only changes the $i$-th bit of $\vec{x}(\vec{b})$ for $\vec{b} \in [0:m]^n$. As such it's obvious that for any $\vec{b} \in [0:m]^n$ and $k > i$ we have $\vec{x_i}(\vec{b})_k = \vec{x}(\vec{b})_k$, as in $\vec{x_i}$ only the first $i$ bits of $\vec{x}(\vec{b})$ have possibly been changed.
    
    Since $L_i$ acts on independent output bits of $\vec{x}$, if $\vec{x}$ was $S$ monotone, it's easy to see $L^i[\vec{x}]$ is $S \cup \{i\}$ monotone. As such $\vec{x_n}$ is $[n]$-monotone, which is equivalent to being coordinate-wise monotone. 

    Finally, note that in $\vec{x_{i}}$ we always have $\vec{x_i}(\vec{b})_{i+1} = \vec{x}(\vec{b})_{i+1}$ by the first point, thus $L^{i+1}[\vec{x_{i}}](\vec{b})_{i+1} = L^{i+1}[\vec{x}](\vec{b})_{i+1}$, since the definition of $L^{i+1}[\vec{x_{i}}]$ doesn't depend on bits $1, \ldots, i$ of $\vec{x}$ output.\Halmos
\end{proof}

We now complete the lemma's proof.
\begin{proof}{Proof of \Cref{lemma: line dist}}
Consider an $i$-line $x_{\alpha, i, \beta}$ of $x$. By definition $L^i[\vec x]_{\alpha, i, \beta}$ is a closest monotone function to $x_{\alpha, i, \beta}$, and hence
\[
|\{b \in [0:m] : L^i[\vec{x}]_{\alpha, i, \beta}(b) \neq x_{\alpha, i, \beta}(b)\}|
=(m+1)d_\Mo(x_{\alpha, i, \beta}).
\]
We have:
    \begin{align*}
        \epsilon \leq d_\CMo(\vec{x}) &\leq d(\vec{x}, L^n[\ldots L^2[L^1[\vec{x}]]\ldots]) \\
                            & \leq \sum_{i=0}^{n-1} d(\vec{x_{i}}, \vec{x_{i+1}} )
    \end{align*} 
    and 
    \begin{align*}
        d(\vec{x_{i}}, \vec{x_{i+1}}) &= \frac{1}{(m+1)^n} |\{ \vec{b} \in [0:m]^n \;:\; \vec{x_i}(\vec{b}) \neq  \vec{x_{i+1}}(\vec{b}) \} | \\ 
        &= \frac{1}{(m+1)^n} |\{ \vec{b} \in [0:m]^n \;:\; \vec{x_i}(\vec{b}) \neq L^{i+1}[ \vec{x_{i}}](\vec{b}) \} | \\ 
        &= \frac{1}{(m+1)^n} |\{ \vec{b} \in [0:m]^n \;:\; \vec{x_i}(\vec{b})_{i+1} \neq L^{i+1}[\vec{x_{i}}](\vec{b})_{i+1} \} | \\ 
        &= \frac{1}{(m+1)^n} |\{ \vec{b} \in [0:m]^n \;:\; \vec{x}(\vec{b})_{i+1} \neq L^{i+1}[\vec{x}](\vec{b})_{i+1} \} | \\ 
        &= \frac{1}{(m+1)^n} \sum_{\substack{\alpha \in [0:m]^i, \\ \beta \in [0:m]^{n-i-1} }} |\{ b \in [0:m] \;:\; \vec{x}(\alpha , b  , \beta )_{i+1} \neq L^{i+1}[\vec{x}]( \alpha , b  , \beta )_{i+1} \} | \\
        &= \E_{\alpha, \beta } \frac{1}{m+1} |\{ b \in [0:m] \;:\; \vec{x}(\alpha , b  , \beta )_{i+1} \neq L^{i+1}[\vec{x}]( \alpha , b  , \beta )_{i+1}  \} | \\
        &= \E_{\alpha, \beta } \frac{1}{m+1} |\{ b \in [0:m] \;:\; x_{\alpha, i+1, \beta}(b) \neq L^{i+1}[\vec{x}]_{ \alpha,i+1, \beta }(b) \} | \\
        &=  \E_{\alpha, \beta} d_\Mo( x_{\alpha, i+1, \beta})
    \end{align*}

Thus, we have \begin{align*} \epsilon \leq \sum_{i=0}^{n-1}  \E_{\alpha, \beta} d_\Mo( x_{\alpha, i+1, \beta}) = n \E_{\alpha, i, \beta} d_\Mo( x_{\alpha, i, \beta}).\Halmos \end{align*}
\end{proof}

\subsection{Proof of \Cref{thm:ub}}
\label{app:earliertester}

First, we will need the following result from \citet{goldreich:LIPIcs.APPROX-RANDOM.2014.704}. 
\begin{lemma}[\cite{goldreich:LIPIcs.APPROX-RANDOM.2014.704}]\label{lemma:work-invest}
    Let $X$ be a random variable that takes values in $[0, 1]$. Suppose $\E\left[ X \right] > \epsilon$. Let $\ell= \left\lceil \log_2 \frac{2}{\epsilon} \right\rceil$. Then, there exists $j \in [\ell]$ such that $\Pr[X > 2^{-j}] \geq 2^j\epsilon/4\ell$.
\end{lemma}

We will also need the following two results, proofs for which can be found in \Cref{app:proof lemma line dist DSIC,app:proof lemma line DSIC tester richer alloc} respectively.
\begin{lemma} \label{lemma: line dist DSIC}
    Consider a function $\vec a: [0:m]^n \to [0:d]^n \times [0,dm]^n$, and a distribution $\mathcal{D}$ over $[0:m]^n$. If $d^\mathcal{D}_{\DSIC}(\vec a) \geq \epsilon$, then

    \begin{align*} \E_{(\alpha,i, \beta) \sim \mathcal{D}_L }\left[ d_{\DSIC}^{\mathcal{D}_{\alpha,i, \beta}}(a_{\alpha, i, \beta}) \right] \geq \epsilon/n. \end{align*}
\end{lemma}

\begin{lemma}\label{lemma: line DSIC tester richer alloc}
Given a function $a:[0:m] \to [0:d] \times [0,dm]$, sample access to a distribution $\mathcal{D}'$ over $[0:m]$, and $0<\epsilon\leq1$, \Cref{algo:mid_tester} is an adaptive one-sided tester with query complexity
\[
\Omic\left(\min\{m,d\log(m)\}+1/\epsilon\right)
\]
to test if $a$ is $\DSIC$ or $d^{\mathcal{D}'}_\DSIC(a) \geq \epsilon$. The tester uses $\lceil2/\epsilon\rceil$ samples from $\mathcal D'$ and runs in time $\Tilde{\Omic}(nd/\epsilon)$ under the unit-cost convention used in \Cref{thm:ub}.
\end{lemma}

Our algorithm and the proof of its correctness follow the one of \citet{goldreich:LIPIcs.APPROX-RANDOM.2014.704} for what the author calls the concatenation problem. However, a key difference in our algorithm is that it picks the instances to test according to the distribution $\mathcal{D}_L$ and not uniformly at random. 

\begin{proof}{Proof of \Cref{thm:ub}}
    Recall that the tester we consider is given in \Cref{algo:newpricetester}. If $\Vec{a}$ is $\DSIC$, then by definition, each line $a_{\alpha,i, \beta}$ is $\DSIC$. So, $T$ (which is a one-sided tester) never rejects at any iteration and the tester accepts $\Vec{a}$.

    Otherwise, we fail to reject $\Vec{a}$ if all calls to $T$ accept. Combining \Cref{lemma: line DSIC tester richer alloc} and \Cref{lemma:work-invest} with our choice of $\epsilon'=\epsilon/2n$, we see that there exists some $j \in [\ell]$ such that 
    \begin{align*}\Pr_{(\alpha,i, \beta) \sim \mathcal{D}_L}[ d_{\DSIC}^{\mathcal{D}_{\alpha,i,\beta}}(a_{\alpha,i,\beta}) \geq 2^{-j} ] \geq 2^j\epsilon'/4\ell \,. \end{align*}
    For this $j$, the algorithm samples at least $12\ell/(2^j\epsilon')$ triples from $\mathcal D_L$. Thus the probability that every corresponding line is at distance $<2^{-j}$ is at most

    \begin{align*} (1- 2^j\epsilon'/4\ell) ^{ 3 \times 4\ell/2^j\epsilon'} \leq 1/10\,.\end{align*}
    
    Hence, with probability at least $9/10$, the algorithm samples a line $a_{\alpha,i,\beta}$ with $d_{\DSIC}^{\mathcal{D}_{\alpha,i,\beta}}(a_{\alpha,i,\beta}) \geq 2^{-j}$. The two independent runs of \Cref{algo:mid_tester} on that line both accept with probability at most $(1/3)^2$.
    
    So, the algorithm rejects $\Vec{a}$ with probability at least $9/10*8/9 \geq 2/3$.

    We now look at the query complexity $q$ of our tester. A call to \Cref{algo:mid_tester} with distance parameter $2^{-j}$ requires $\Omic(M+2^j)$ queries, where
    \[
    M=\min\{m,d\log(m)\}.
    \]
    Since $2^\ell<4/\epsilon'$, every unrounded repetition count $12\ell/(2^j\epsilon')$ is greater than $3\ell$. Thus the ceilings change the following bounds by at most a constant factor.

    \begin{align*}
       q&\leq \sum_{j=1}^\ell \Omic\left(\frac{\ell}{2^j\epsilon'}\right) \Omic \left( M+2^j \right) \\
        &= \Omic\left( \frac{M \ell}{ \epsilon'}\right) \times \sum_{j=1}^\ell \frac{1}{2^j} + \sum_{j=1}^\ell  \Omic\left( \frac{\ell}{2^j\epsilon'} \times \frac{1}{2^{-j}}\right) \\
        & \leq  \Omic\left( \frac{M\ell}{ \epsilon' }\right) + \sum_{j=1}^\ell  \Omic\left( \frac{\ell}{\epsilon'}\right) \\
        &= \Omic\left( \frac{M \ell}{ \epsilon'} + \frac{\ell^2}{\epsilon'} \right) \\
        &= \Omic\left( \frac{n}{\epsilon} \log\left(\frac{2n}{\epsilon}\right) M +  \frac{n}{\epsilon} \log^2\left(\frac{2n}{\epsilon}\right) \right)\,.
    \end{align*}

    Finally, each call to \Cref{algo:mid_tester} with distance parameter $2^{-j}$ uses $2^{j+1}$ samples from $\mathcal{D}_{\alpha,i,\beta}$. Hence the total number of conditional samples is

    \begin{align*}
    \sum_{j=1}^\ell \Omic\left(\frac{\ell}{2^j\epsilon'}\right)(1+\Omic(2^j))
        &= \Omic\left(\frac{\ell}{\epsilon'}+\frac{\ell^2}{\epsilon'}\right)\\
        &= \Omic\left(\frac{n}{\epsilon}\log^2\left(\frac{2n}{\epsilon}\right)\right)\,.
    \end{align*}

    We use exact, unit-cost arithmetic for reported prices and charge $\Tilde{\Omic}(n)$ time for each query to $\Vec a$ and each sample from $\mathcal D$ or a conditional distribution $\mathcal D_{\alpha,i,\beta}$. A call to $T$ with distance parameter $2^{-j}$ takes time $\Tilde{\Omic}(2^jnd)$. Thus the running time is
    \begin{align*}
    \sum_{j=1}^\ell \Omic\left(\frac{\ell}{2^j\epsilon'}\right)
    \left(n+\Tilde{\Omic}\left(2^jnd\right)\right)
    &=\Tilde{\Omic}\left(\frac{\ell nd}{\epsilon'}\sum_{j=1}^\ell 2^{-j}2^j\right)\\
        &= \Tilde{\Omic}\left(\ell/\epsilon' \times nd \times \ell  \right) \\
        &= \Tilde{\Omic}( \frac{n\log(2n/\epsilon)}{\epsilon}  \times nd \times \log(2n/\epsilon)) \\
        &= \Tilde{\Omic}\left( \frac{n^2d}{\epsilon} \right)\,.\Halmos
    \end{align*}
    
\end{proof}

\subsection{Proof of \Cref{lemma: line dist DSIC}}
\label{app:proof lemma line dist DSIC}
Before proving the lemma, we will need the following result.
\begin{lemma}\label{lem: rexpress distribution}
    
Let $\Vec{b} \in [0:m]^n$, fix $i \in [n]$, and let $\alpha$ denote the first $i-1$ coordinates of $\Vec{b}$ and $\beta$ the last $n-i$ coordinates of the vector. We have
\begin{align*}\Pr_{X \sim \mathcal{D}_L}[X=(\alpha,i,\beta)] \times \Pr_{B \sim \mathcal{D}_{\alpha,i,\beta}}[B=\Vec{b}_i] = \frac{1}{n} \Pr_{Z \sim \mathcal{D}}[Z = \vec{b}].  \end{align*}
\end{lemma}
\begin{proof}{Proof}
By the definition of $\mathcal{D}_L$, we have
\begin{align*}
     \Pr_{X \sim \mathcal{D}_L}[X=(\alpha,i,\beta)] \times \Pr_{B \sim \mathcal{D}_{\alpha,i,\beta}}[B=\Vec{b}_i] &=  \left( \frac{1}{n} \sum_{b=0}^m \Pr_{X \sim \mathcal{D}}[ X =  \alpha , b , \beta] \right) \Pr_{B \sim \mathcal{D}_{\alpha,i,\beta}}[B=\Vec{b}_i].
\end{align*}

If  $\sum_{b=0}^m \Pr_{Y \sim \mathcal{D}}[ Y =  \alpha , b , \beta] \neq 0$ we have :

\begin{align*}
     \left(\frac{1}{n} \sum_{b=0}^m \Pr_{Y \sim \mathcal{D}}[ Y =  \alpha , b , \beta] \right) \Pr_{B \sim \mathcal{D}_{\alpha,i,\beta}}[B=\vec{b}_i] &= \left(\frac{1}{n} \sum_{b=0}^m \Pr_{Y \sim \mathcal{D}}[ Y =  \alpha , b , \beta] \right) \dfrac{ \Pr\limits_{Z \sim \mathcal{D}}\left[ Z =  \alpha , \Vec{b}_i , \beta \right]}{\sum\limits_{b'=0}^{m} \Pr\limits_{Z \sim \mathcal{D}}\left[ Z =  \alpha , b' , \beta\right]} \\
    &=  \frac{1}{n} \Pr_{Z' \sim \mathcal{D}}\left[Z'= \alpha , \Vec{b}_i , \beta\right] \\
    &= \frac{1}{n} \Pr_{Z \sim \mathcal{D}}[Z= \Vec{b}]\,.
\end{align*}

If $\sum_{b=0}^m \Pr_{Y \sim \mathcal{D}}[ Y =  \alpha , b , \beta] = 0$, which implies $\Pr_{Z \sim \mathcal{D}}[Z = \vec{b}]=0$, we get 
\begin{align*}
     \left(\frac{1}{n} \sum_{b=0}^m \Pr_{Y \sim \mathcal{D}}[ Y =  \alpha , b , \beta] \right) \Pr_{B \sim \mathcal{D}_{\alpha,i,\beta}}[B=\vec{b}_i] = 0 = \frac{1}{n} \Pr_{Z \sim \mathcal{D}}[Z= \Vec{b}] \,.\Halmos
\end{align*}
\end{proof}

\begin{proof}{Proof of \Cref{lemma: line dist DSIC}}

Consider a function $\vec{a}=(\vec{x}, \vec p): [0:m]^n \to [0:d]^n \times [0,dm]^n$ and a distribution $\mathcal{D}$ over $[0:m]^n$. For $i \in [n], \alpha \in [0:m]^{i-1}, \beta \in [0:m]^{n-i}$, let $(\Tilde{x}^{\alpha,i,\beta}, \Tilde{p}^{\alpha,i,\beta})$ be the $\DSIC$ function closest to $a_{\alpha,i,\beta}$ under the distribution $\mathcal{D}_{\alpha,i,\beta}$. That is, 
\begin{align*} \Pr_{b \sim \mathcal{D}_{\alpha,i, \beta}}\left[\Tilde{x}^{\alpha,i,\beta}(b) \neq x_{\alpha,i, \beta}(b) \lor   \Tilde{p}^{\alpha,i,\beta}(b) \neq p_{\alpha,i, \beta}(b)\right] = d_{\DSIC}^{\mathcal{D}_{\alpha,i,\beta}}(a_{\alpha,i,\beta})   \,. \end{align*}

We then define the functions $\vec{y} : [0:m]^n \to [0:d]^n, \Vec{q} : [0:m]^n \to [0,dm]^n$ such that $(\Vec{y}(\alpha,b, \beta)_i,\Vec{q}(\alpha,b, \beta)_i) = (\tilde{x}^{\alpha,i,\beta}(b), \tilde{p}^{\alpha,i,\beta}(b))$ for all $b \in [0:m]$.

Note that this is a valid definition for $\Vec{y}, \Vec{q}$ since each bit of the output is defined exactly once. Moreover, it's easy to see that $\Vec{y}, \Vec{q}$ is $\DSIC$ since all the lines of $\Vec{y}$ are coordinate wise monotone, and each of $\Vec{q}$ is a valid price function for the corresponding line of $\Vec{y}$.

Since $d^\mathcal{D}_{\DSIC}(\Vec a)\geq\epsilon$ and $(\Vec y, \Vec q)$ is $\DSIC$, we have that
\begin{align*}
    \epsilon &\leq \Pr_{\Vec{b} \sim \mathcal{D}}\left[ \vec{a}(\Vec{b}) \neq \left(\Vec{y}(\Vec{b}),  \Vec{q}(\vec b) \right) \right] \\    
    &= \sum_{\Vec{b} \in [0:m]^n} \Pr_{Z \sim \mathcal{D}}[Z=\Vec{b} ] \times \mathds{1}\left[\vec{a}(\Vec{b}) \neq \left(\Vec{y}(\Vec{b}),  \Vec{q}(\vec b) \right)\right] \\
    &\leq  \sum_{\Vec{b} \in [0:m]^n} \Pr_{Z \sim \mathcal{D}}[Z=\Vec{b} ] \times \sum_{i=1}^n \mathds{1}\left[\vec{a}(\Vec{b})_i \neq \left(\Vec{y}(\Vec{b})_i,  \Vec{q}(\vec b)_i \right)\right] \\
    &=  \sum_{i=1}^n \sum_{\Vec{b} \in [0:m]^n} \Pr_{Z \sim \mathcal{D}}[Z=\Vec{b} ] \times \mathds{1}\left[\vec{a}(\Vec{b})_i \neq \left(\Vec{y}(\Vec{b})_i,  \Vec{q}(\vec b)_i \right)\right]  \\
     &= \sum_{i=1}^n \sum_{\substack{\alpha \in [0:m]^{i-1},\\ \beta \in [0:m]^{n-i} }} \sum_{b \in [0:m]} \Pr_{Z \sim \mathcal{D}}[Z= \alpha , b , \beta ] \times \mathds{1}\left[\vec{a}(\alpha , b , \beta )_i \neq \left(\Vec{y}(\alpha , b , \beta )_i,  \Vec{q}(\alpha , b , \beta )_i \right) \right] \,.
\end{align*}
\clearpage
Finally, using the above, we derive the lemma:
{\allowdisplaybreaks
\begin{align*}
         \E_{(\alpha,i, \beta) \sim \mathcal{D}_L }\left[ d_{\DSIC}^{\mathcal{D}_{\alpha,i, \beta}}(a_{\alpha, i, \beta}) \right] &= \sum_{i=1}^n \sum_{\substack{\alpha \in [0:m]^{i-1},\\ \beta \in [0:m]^{n-i} }} \Pr_{X \sim \mathcal{D}_L}\left[X= (\alpha, i, \beta) \right]d_{\DSIC}^{\mathcal{D}_{\alpha,i, \beta}}(\Vec{a}_{\alpha, i, \beta}) \\
         &= \sum_{i=1}^n \sum_{\substack{\alpha \in [0:m]^{i-1},\\ \beta \in [0:m]^{n-i} }} \Pr_{X \sim \mathcal{D}_L}[X= (\alpha, i, \beta) ]\, \times \\ & \quad \quad \quad \quad \quad \quad  \quad \quad \quad\Pr_{b \sim \mathcal{D}_{\alpha, i, \beta}}\left[\left(\Tilde{x}^{\alpha,i,\beta}(b),\Tilde{p}^{\alpha,i,\beta}(b)\right)  \neq a_{\alpha,i, \beta}(b) \right] \\
         &= \sum_{i=1}^n \sum_{\substack{\alpha \in [0:m]^{i-1},\\ \beta \in [0:m]^{n-i} }} \Pr_{X \sim \mathcal{D}_L}[X= (\alpha, i, \beta) ] \, \times \\ & \quad \quad \quad \sum_{b \in [0:m]} \Pr_{B \sim \mathcal{D}_{\alpha, i, \beta}}[B=b] \mathds{1}\left[\left(\Tilde{x}^{\alpha,i,\beta}(b),\Tilde{p}^{\alpha,i,\beta}(b)\right)  \neq a_{\alpha,i, \beta}(b) \right]\\
         &= \sum_{i=1}^n \sum_{\substack{\alpha \in [0:m]^{i-1},\\ \beta \in [0:m]^{n-i} }} \sum_{b \in [0:m]}  \Pr_{X \sim \mathcal{D}_L}[X= (\alpha, i, \beta) ] \Pr_{B \sim \mathcal{D}_{\alpha, i, \beta}}[B=b]  \, \times \\&  \quad \quad \quad \quad \quad \quad \quad \quad \quad \mathds{1}\left[\left(\Tilde{x}^{\alpha,i,\beta}(b),\Tilde{p}^{\alpha,i,\beta}(b)\right)  \neq a_{\alpha,i, \beta}(b) \right] \\
         &= \sum_{i=1}^n \sum_{\substack{\alpha \in [0:m]^{i-1},\\ \beta \in [0:m]^{n-i} }} \sum_{b \in [0:m]}  \frac{1}{n} \Pr_{Z \sim \mathcal{D}}[Z=\alpha , b , \beta] \times \\ & \quad \quad \quad \quad  \quad  \quad \quad \quad \quad   \mathds{1}\left[\left(\Tilde{x}^{\alpha,i,\beta}(b),\Tilde{p}^{\alpha,i,\beta}(b)\right)  \neq a_{\alpha,i, \beta}(b) \right]\\
          &= \frac{1}{n} \sum_{i=1}^n \sum_{\substack{\alpha \in [0:m]^{i-1},\\ \beta \in [0:m]^{n-i} }} \sum_{b \in [0:m]} \Pr_{Z \sim \mathcal{D}}[Z= \alpha , b , \beta ] \, \times \\& \quad \quad  \quad \quad \quad \quad \quad \quad \quad \mathds{1}\left[\vec{a}(\alpha , b , \beta )_i \neq \left(\Vec{y}(\alpha , b , \beta )_i,  \Vec{q}(\alpha , b , \beta )_i \right) \right] \\
         &\geq \epsilon/n \,,
     \end{align*}
}
where we used \Cref{lem: rexpress distribution} in the fourth line.\Halmos
\end{proof}

\subsection{Proof of \Cref{lemma: line DSIC tester richer alloc}}
\label{app:proof lemma line DSIC tester richer alloc}

\begin{proof}{Proof}
The tester we consider is the one given by \Cref{algo:mid_tester}. The idea of the algorithm is the following: Given query access to $a=(x,p): [0:m] \to [0:d] \times [0,dm]$, we try to find all the "jumps" in the allocation function $x$. That is, for each $\ell\in [0:d]$ we look for the smallest bid $b$ such that $x(b)=\ell$. We record these bids in a set $J_a$. We first check if the points in $J_a$ reveal $x$ isn't monotone (we always have $t_i < t_{i+1}$). And that at each jump point, the increase in the price is "possible" ($t_i-1 \leq h_i \leq t_{i}$). 

Finally, we build a $\DSIC$ function that agrees with $a$ on $J_a$. In particular, if $a$ is $\DSIC$ this function will exactly be $a$. Otherwise, we know that $a$ and the built function must disagree on a randomly chosen point with high probability.

The binary searches use at most $\Omic(\min\{m,d\log(m)\})$ queries: answers are cached, so the search stops costing new queries once all $m+1$ points have been queried. Querying any jump points not already cached costs at most $d+1$ additional queries, which is absorbed by this bound. The last loop uses $\lceil2/\epsilon\rceil$ queries.

It remains to prove this tester is indeed one-sided. Since $s_0$ is the smallest allocation for which we found a bid $b$ with $x(b)=s_0$, we must have $v_{s_0}=t_0=0$, unless $x$ isn't monotone. 
Assume $a=(x,p)$ is $\DSIC$, then $x$ is monotone. Then, for $\ell \in [0:d]$, such that there exists a bid $b$ with $x(b)=\ell$, we will add $\ell$ to $S$. Furthermore, for $s_i$ in $S$, $t_i$ is always the smallest bid with $x(t_i)=s_i$, and we have $J_a=J_x$. Since $x$ is monotone, we always have $t_i \leq t_{i+1}$. Because of how we define $y$, it's not hard to see that we have $y=x$. We then define $h_i = \left(p(t_i)-p(t_{i-1})\right) / \left(x(t_i)-x(t_{i-1})\right)$ and build the function $p^y$. It's easy to see that $p^y$ is a valid price function for $y$ and thus also for $x$. However, note that since $p$ is a valid price for $x$, it is implicitly defined by values $h_1', \ldots, h_k'$ such that for all $b \in [0:m]$, we have
\begin{align*}
p(b) = 0 + \sum_{j=1}^{i^b} h'_j \times (x(t_j) - x(t_{j-1}))\,.
\end{align*}
But, because of how $h_j$ is computed we have $h_j=h'_j$ for all $j$. So $p^y=p^x$. So we if $a=(x,p)$ is $\DSIC$ then $x=y$ and $p=p^y$, so it's clear we will accept $a$. 

Now assume $a$ is $\epsilon$ far from $\DSIC$ under $\mathcal{D}'$. If we haven't rejected $a$ by the end of the second for loop, we consider the function $(y,p^y)$. Since $y$ is monotone and $p^y$ is a valid price function for $y$, this is a $\DSIC$ function. However, since $d_{\DSIC}^{\mathcal{D}'}(a) \geq \epsilon$ this implies 

\begin{align*} \Pr_{b \sim \mathcal{D}'}[ y(b) \neq x(b) \lor p(b) \neq p^y(b)] \geq \epsilon \,. \end{align*} We sample $b$ according to $\mathcal{D}'$ and check this condition $\lceil2/\epsilon\rceil$ times, so the tester fails to reject $a$ with probability at most $(1-\epsilon)^{\lceil2/\epsilon\rceil} \leq e^{-2}<1/3$. Hence, this is a one-sided adaptive $\DSIC$ tester under $\mathcal{D}'$.

Recall that we assume sampling from $\mathcal{D}'$ and querying $a$ take time $\Tilde{\Omic}(n)$. 

Hence the first for loop, where perform the binary search phase takes time $\Tilde{\Omic}(nd)$. Sorting $S$ takes time $\Tilde{\Omic}(d)$. The second for loops takes time $|S|=\Omic(d)$.

One doesn't have to write down $y,p^y$ for all $b \in [0:m]$. Given $b$, one can perform binary search in $J_a$ (assuming it's sorted) to find $t_b$. Then $p^y(b)$ is the sum of at most $i^b \leq |J_a|$ many values, all of which are already known. So computing $y(b), p^y(b)$ would take time $\Omic(d)$. 

So each time we sample a point, it takes $\Omic(d)$ time to compute $(x(b),p(b))$ and it took $\Tilde{\Omic}(n)$ to query $a(b)$ and sample $\mathcal{D}'$. So the last for loop takes time $\Tilde{\Omic}(nd/\epsilon)$.

So the runtime of the algorithm is $\Tilde{\Omic}\left(nd/\epsilon \right)$.\Halmos
\end{proof}

\subsection{Proof of \Cref{theorem:tree_DSIC} and \Cref{theorem:tree_not_DSIC}}
\label{app:tree_lemmas}

First, we will work with an equivalent definition of $\DSIC$, which will make the theorems of this section easier to prove.

\begin{definition}[$\DSIC$ function]\label{def:p-DSIC}
    $a=(x,p) : [0:m] \to [0:d] \times [0:dm]$ is $\DSIC$ if $p(0)=0$ and, for all $b \in [m]$:
    \begin{itemize}
        \item $x(b-1) \leq x(b)$
        \item $p(b)=p(b-1)+h_{b} \times \left(x(b)-x(b-1)\right)$ where $b-1 \leq h_{b} \leq b$. 
    \end{itemize}
\end{definition}

We will consider the tree $T_a$ as we defined it in \Cref{def: tree T_a}. Here are some simple observations:
\begin{observation}\label{obs:parent_empty}
    Consider the tree $T_a$ for some function $a : [0:m] \to [0:d] \times [0,dm]$. 
    
    \begin{itemize}
        \item For any node $u$ with parent $p$ if $\mathcal{P}_{u}=\emptyset$ then $\mathcal{T}_{p}=\emptyset$.
        \item  For an internal node $u$, if $\mathcal{T}_u=\emptyset$, then for any of its successors $v$ we have $\mathcal{P}_v = \emptyset$. If $v$ is internal, this also means $\mathcal{T}_v=\emptyset$.
        \item Consider an internal node $u=[s:t]$. Let $b \in \{\frac{s+t-1}{2}, \frac{s+t+1}{2}\}$, then for any of $u$'s successor $v$, since $\mathcal{P}_{v} \subseteq \mathcal{T}_u$, we have that if $f \in \mathcal{P}_{v}$ then $f(b)=a(b)$.
    \end{itemize}
\end{observation}

We begin with the proof of \Cref{theorem:tree_DSIC}.
\begin{proof}{Proof}

Assume that $a$ is DSIC. We prove by induction on the depth of
the nodes that $a\in\mathcal P_u$ for every node $u$, and that
$a\in\mathcal T_u$ whenever $u$ is internal.

For the root $r=[0:m]$, we have $a\in\mathcal P_r$ because $a$
is DSIC and agrees with itself at $0$ and $m$. Moreover,
$a\in\mathcal T_r$, since $a$ also agrees with itself at the
two middle points of $r$.

Now let $u=[s:t]$ be a non-root node with parent $p$.
By the induction hypothesis, $a\in\mathcal T_p$.
Since $\mathcal P_u=\mathcal T_p$, we obtain
$a\in\mathcal P_u$. If $u$ is internal, the set
$\mathcal T_u$ consists of the functions in $\mathcal P_u$
that agree with $a$ at
\[
    \frac{s+t-1}{2}
    \qquad\text{and}\qquad
    \frac{s+t+1}{2}.
\]
The function $a$ itself satisfies these two conditions, so
$a\in\mathcal T_u$ as well. This completes the induction. Thus $\mathcal P_u\neq\varnothing$ for every node $u$, and
$\mathcal T_u\neq\varnothing$ for every internal node $u$. \Halmos
\end{proof}

We will now turn to the proof of \Cref{theorem:tree_not_DSIC}, which is more involved. We will first prove the following lemma.
\begin{lemma}\label{lemma:bad_set}
    Consider a function $a:[0:m]\to [0:d] \times [0,dm]$. Assume the root of $T_a$ is such that $\mathcal{P}_{[0:m]} \neq \emptyset$. Then let 
    \begin{align*}\mathcal{B} := \{ u \;|\; \text{$u$ isn't a leaf and } \mathcal{P}_{u} \neq \emptyset \text{ and } \mathcal{T}_{u}=\emptyset \} \,.\end{align*}

    \begin{itemize}
        \item If $u, u'$ are in $\mathcal{B}$ then $u$ isn't an ancestor of $u'$. 
        \item $b \in [0:m]$ is such that $\mathcal{P}_{[b:b]}=\emptyset$ if and only if there exists a \emph{unique} node $u \in \mathcal{B}$ such that $b \in u$.   
    \end{itemize}
    
\end{lemma}

\begin{proof}{Proof}

    First, let $u, u'$ be two nodes in $\mathcal{B}$. Assume for contradiction that $u$ is an ancestor of $u'$. Then, we have that $\mathcal{T}_u=\emptyset$. So by \Cref{obs:parent_empty} we have that $\mathcal{P}_v=\emptyset$ for any successor $v$ of $u$. But $u'$ is in $\mathcal{B}$, so $\mathcal{P}_{u'} \neq \emptyset$, which contradicts the fact $u'$ is a successor of $u$.

    Note $b$ be such that $\mathcal{P}_{[b:b]}=\emptyset$, we first prove that $[b:b]$ has an ancestor $u$, meaning $b \in u$ such that $u \in B$.
    
    By \Cref{obs:parent_empty}, we know that for the parent $p$ of $b$ we must have $\mathcal{T}_{p}=\emptyset$. Hence, either $p \in \mathcal{B}$, or $\mathcal{P}_p = \emptyset$. In the first case, we're done. Else we can look at the parent of $p$. Since the root is such that $\mathcal{P}_{[0:m]} \neq \emptyset$, we're guaranteed to hit an ancestor $v$ of $b$ that we can put $v$ in $\mathcal{B}$, and we have $b \in v$.

    If $b$ had two ancestors $u,u' \in \mathcal{B}$. Then without loss of generality $u$ is an ancestor of $u'$, which is contradiction.

    Finally for the other direction, if $b \in u$ and $u \in \mathcal{B}$, then by \Cref{obs:parent_empty} we have that for the successors $v$ of $u$ have $P_{v}=\emptyset$. So, $P_{[b:b]}=\emptyset$.\Halmos
\end{proof}

We now turn to the proof of \Cref{theorem:tree_not_DSIC}.
\begin{proof}{Proof}
    Let $d_{\DSIC}^\mathcal{D}(a) \geq \epsilon$,  let $P=\Pr_{b \sim \mathcal{D}}\left[ \mathcal{P}_{[b:b]}=\emptyset \right]$ and assume for contradiction $P < \epsilon.$

First, we can't have that for the root $[0:m]$ has $P_{[0:m]}=\emptyset$. Indeed, all the leaves of the tree would have $\mathcal{P}_{[b:b]}=\emptyset$. Which would contradict our assumption since $1=P<\epsilon \leq 1$.

    Hence, consider the set $\mathcal{B}$ as defined in \Cref{lemma:bad_set}. Let 
    \begin{align*}\gamma=\sum_{[s:t] \in \mathcal{B}} \sum_{b \in [s:t]} \Pr_{X \sim \mathcal{D}}[ X=b ] \,.\end{align*}

    By \Cref{lemma:bad_set} if $u=[s:t] \in \mathcal{B}$, then for all $b \in u$, we have $\mathcal{P}_{[b:b]}=\emptyset$. And for $b$ with $\mathcal{P}_{[b:b]}=\emptyset$, there is a unique $u \in \mathcal{B}$ with $b \in u$. Thus, $\gamma = P$.

    We will now construct a $\DSIC$ function $\alpha=(\xi, \phi): [0:m] \to [0:d] \times [0:dm]$ such that $\Pr_{b \sim \mathcal{D}}[\alpha(b)\neq a(b)] \leq \gamma$. Assuming, we can do this, we immediately reach a contradiction since $\gamma=P<\epsilon\leq d_{\DSIC}^\mathcal{D}(a)$.

    For each node $u=[s:t] \in \mathcal{B}$, pick a function $f_u=(x_u, p_u)$ in $\mathcal{P}_u$. Note that $f_u$ is $\DSIC$.
    Finally, we define $\alpha$ as follow :

\begin{equation*}
    \alpha(b) =
    \begin{cases}
      f_u(b) & \text{if } b \in u \text{ for some } u \in \mathcal{B}. \\
      a(b)        & \text{otherwise.} 
    \end{cases}
\end{equation*}

Thus functions is well defined since for every $b$, there is at most one $u \in \mathcal{B}$ with $b \in u$ by \Cref{lemma:bad_set}. We also have 
\begin{align*}
\Pr_{b \sim \mathcal{D}}[\alpha(b)\neq a(b)]& \leq\sum_{[s:t] \in \mathcal{B}} \sum_{b \in [s:t]} \Pr_{X \sim \mathcal{D}}[ X=b ]=\gamma\,. 
\end{align*}

So it remains to prove $\alpha$ is $\DSIC$. We will now prove this using the equivalent definition of $\DSIC$ given in \Cref{def:p-DSIC}.

First, $\phi(0)=0$. There are two cases.
\begin{enumerate}
    \item $\alpha(0)=a(0)$, meaning $\phi(0)=p(0)$. By our assumption $P_{[0:m]}$ isn't empty, so it must be that $p(0)=0$, as otherwise there would be no $\DSIC$ function consistent with $a$ on $b=0$. 
    \item Otherwise there exists $u \in \mathcal{B}$ with $\alpha(0)=f_u(0)$, meaning $\phi(0)=p_u(0)$. Since $f_u$ is $\DSIC$, we must have that $p_u(0)=0$.
\end{enumerate}

We will now show that for any  $b \in [m]$, we have $\xi(b-1) \leq \xi(b)$, and there exists $b-1 \leq h_b \leq b$ such that $\phi(b)=\phi(b-1)+h_b\times\left( \xi(b)-\xi(b-1)\right)$. 

 \begin{itemize}
    \item If $b, b-1$ are in the same interval $u \in \mathcal{B}$.  Then $\alpha(b-1)=f_{u}(b-1)$, $\alpha(b)=f_{u}(b)$. Since, $f_{u}$ is $\DSIC$ we must have that  $\xi(b)=x_u(b)\geq x_u(b-1)=\xi(b-1)$. We also have there must exist $ b-1 \leq h_b \leq b$ such that 
    \begin{align*}\phi(b)=p_u(b)=p_u(b-1)+ h_b \times\left( x_u(b)-x_u(b-1)\right)   =\phi(b-1)+h_b \times\left( \xi(b)-\xi(b-1)\right). \end{align*}

    \item Else, let $u=[s:t]$ be the last interval in $T_a$ containing both $b, b-1$. In particular, it must be that $b,b-1$ are the "middle points" of $[s:t]$, that is $b-1=\frac{s+t-1}{2}$, $b=\frac{s+t+1}{2}$. Since $b,b-1$ aren't in the same interval of $\mathcal{B}$, it must be that $\mathcal{P}_u, \mathcal{T}_u \neq \emptyset$. 
Indeed, suppose first that $\mathcal P_u=\varnothing$.
Since $\mathcal P_{[0:m]}\neq\varnothing$, let $w$ be the
first node on the path from the root to $u$ for which
$\mathcal P_w=\varnothing$, and let $v$ be its parent.
Then $\mathcal P_v\neq\varnothing$ and
$\mathcal T_v=\mathcal P_w=\varnothing$, so $v\in\mathcal B$.
But $v$ is an ancestor of $u$ and therefore contains both
$b-1$ and $b$, contradicting our assumption that these
points do not lie in a common interval of $\mathcal B$.
Hence $\mathcal P_u\neq\varnothing$.
If $\mathcal T_u=\varnothing$, then $u$ itself belongs to
$\mathcal B$, giving the same contradiction.
Thus $\mathcal T_u\neq\varnothing$ as well.
    
    Note that for any function $f$ in $\mathcal{T}_u$, we must have $f(b)=a(b)$ and $f(b-1)=a(b-1)$.

    Let $\beta \in \{b-1,b\}$ :
    \begin{itemize}
        \item If $\beta$ isn't in some interval in $\mathcal{B}$, then $\alpha(\beta)=a(\beta)$.
        \item If $\beta$ is in some interval $k \in \mathcal{B}$, then $\alpha(\beta)=f_{k}(\beta)$, where $k$ is a descendant of $u$. Hence we have $f_k \in \mathcal{P}_{k} \subseteq \mathcal{T}_{u}$. So $f_k(\beta)=a(\beta)$.
    \end{itemize}
    So, we have $\alpha(b)=a(b)$ and $\alpha(b-1)=a(b-1)$.

    Furthermore, since $\mathcal{T}_u \neq \emptyset$, we must $x(b-1) \leq x(b)$ and there exists $b-1 \leq h_b \leq b$ such that $p(b)=p(b-1)+h_b\times\left( x(b)-x(b-1)\right)$, as otherwise there would be no $\DSIC$ functions in $\mathcal{T}_u$.

    So this proves that $\xi(b-1) = x(b-1) \leq x(b) = \xi(b)$. And there is $b-1 \leq h_b \leq b$ with \begin{align*}\phi(b)=p(b)=p(b-1)+h_b\times\left( x(b)-x(b-1)\right)=\phi(b-1)+h_b\times\left( \xi(b)-\xi(b-1)\right) \end{align*}
\end{itemize}
   
This proves $\alpha$ is $\DSIC$. Hence, we must have $\Pr_{b \sim \mathcal{D}}\left[ \mathcal{P}_{[b:b]}=\emptyset \right] \geq \epsilon$.\Halmos

\end{proof}

\subsection{Proof of \Cref{thm:ub_amazing}}
\label{app:proof of amazing theorem} We first present a simple algorithm to check if $\mathcal{P}_{[b:b]}=\emptyset$ with $\Omic(\log(m))$ queries.

\setcounter{equation}{0}
\begin{algorithm}[H]
\caption{Algorithm to check if $\mathcal{P}_{[b^*:b^*]}= \emptyset$.}\label{algo:empty_tester}
\hspace*{\algorithmicindent}\textbf{Input: }{Query access to $a : [0:m] \to [0:d] \times [0,dm]$, a value $b^* \in [0:m]$}
\begin{algorithmic}[1]
\State Let $C=\{0,m\}$, $s=0, t=m$.  
\State Query $a(0), a(m)$. 
 \While{$s \neq t$}
\State Add $\frac{s+t-1}{2}, \frac{s+t+1}{2}$ to $C$ and query $a(\frac{s+t-1}{2}), a(\frac{s+t+1}{2})$.  
\If{$b^* \leq \frac{s+t-1}{2}$}
\State Set $t=\frac{s+t-1}{2}$. 
\Else 
\State Set $s= \frac{s+t+1}{2}$.
\EndIf
\EndWhile
\For{each function $g : [0:m] \to [0:d]$}
\If{$g$ isn't monotone}
\State Continue.
\EndIf
\State Check if the following system has a solution:
\Statex{\centering
\fbox{\parbox{\textwidth-20\fboxsep-2\fboxrule}{
\text{\textbf{Variables} : $h_b\, \forall\, b \in [m]\,, p_b\, \forall\, b \in [0:m]$}\\
\textbf{Constraints :}
\begin{align}
        & b-1 \leq h_{b} \leq b & \forall b \in [m] \\
        & 0 \leq  p_b  \leq dm & \forall b \in [m] \\
        &  p_b=p_{b-1}+ h_{b} \times \left(g(b)-g(b-1)\right)  \:\:\: & \forall b \in [m] \\
        &  p_0 = 0 & \\
        &  (g(b),p_b)=a(b) & \forall b \in C
\end{align}
}}}
\If{a solution exists}
\State Output $\mathcal{P}_{[b^*:b^*]}\neq \emptyset$.
\EndIf
\EndFor
\State Output $\mathcal{P}_{[b^*:b^*]}= \emptyset$. 
\end{algorithmic}
\end{algorithm}

\begin{theorem}
    Given $b^* \in [0:m]$ and query access to $a : [0:m] \to [0:d] \times [0,dm]$, \Cref{algo:empty_tester} correctly checks if $\mathcal{P}_{[b^*:b^*]}$ is empty in $T_a$. The algorithm uses $\Omic(\log(m))$ queries and runs in time $\Tilde{\Omic}(n)+d^{{\Omic(m)}}$ in the unit-cost model.
\end{theorem}

\begin{proof}{Proof}
    For any $[s:t]$ that is an ancestor of $[b^*:b^*]$ in $T_a$, we have  $\mathcal{P}_{[b^*:b^*]} \subseteq \mathcal{T}_{[s:t]}$ and we also have $\mathcal{P}_{[b^*:b^*]} \subseteq \mathcal{P}_{[0:m]}$. In particular, $\mathcal{P}_{[b^*:b^*]}$ is the set of $\DSIC$ functions that are consistent with $a(0), a(m)$ and  $a(\frac{s+t-1}{2}), a(\frac{s+t+1}{2})$ for all nodes $[s:t]$ that are ancestors of $[b^*:b^*]$ in $T_a$.

    Hence, the algorithm starts by querying $a$ on $0$ and $m$. Then the algorithm goes from $[0:m]$ to $[b:b]$ in $T_a$ and each time it's at interval $[s:t]$, the algorithms queries $a$ on $\frac{s+t-1}{2}$ and $\frac{s+t+1}{2}$. Which takes $\Omic(\log(m))$ many queries and takes time $\Tilde{\Omic}(n)$.

    Now to check if the set $\mathcal{P}_{[b^*:b^*]}$ is empty, we look at all possible allocations $g : [0:m] \to [0:d]$. If the allocation isn't monotone, we reject since no $\DSIC$ functions would have such an allocation scheme. 

    Hence, we check if it's possible to build a $\DSIC$ function $\alpha(b)=(g(b),p_b)$ that is $\DSIC$ and consistent with $a$ on the points in $C$. Constraints $(1)$ to $(4)$ correspond to enforcing prices to be well-defined. Constraint $(5)$ ensures that the function agrees with $a$ on the points in $C$.

    If the constraints have a solution then clearly $\mathcal{P}_{[b^*:b^*]} \neq \emptyset$. If for all functions $g$, $g$ isn't monotone or the constraints have no solution we have that $\mathcal{P}_{[b^*:b^*]} = \emptyset$.

    The while loop takes $\Tilde{\Omic}(n)$ time under the oracle-cost convention. The algorithm then loops over all $(d+1)^{m+1}$ functions $g : [0:m] \to [0:d]$. For each one, checking monotonicity and linear feasibility takes time polynomial in $m$.

    Thus the overall running time is $\Tilde{\Omic}(n)+d^{{\Omic}(m)}$.\Halmos
\end{proof}

With the above algorithm, we can build a tester with query complexity $\Omic(\log(m)/\epsilon)$ to test if $a$ is $\DSIC$ or $d_{\DSIC}^{\mathcal{D}'}(a) \geq \epsilon$ for an arbitrary distribution $\mathcal{D}'$:

\begin{algorithm}[H] 
\caption{One-sided non-adaptive $\DSIC$ tester for $a:[0:m] \to [0:d] \times [0,dm]$ under distribution $\mathcal{D}'$. }\label{algo:best_algo}
\hspace*{\algorithmicindent} \textbf{Input: }{Distance parameter $0<\epsilon\leq1$, sample access to a distribution $\mathcal{D}'$ over $[0:m]$ and query access to $a=(x,p): [0:m] \to [0:d] \times [0,dm]$.}
 \begin{algorithmic}[1]
\For{ $j=1$ to $\lceil2/\epsilon\rceil$ }
\State Sample a point $b$ from $\mathcal{D}'$ and use \Cref{algo:empty_tester} to check if $\mathcal{P}_{[b:b]}=\emptyset$. 
\If{$\mathcal{P}_{[b:b]}=\emptyset$}
\State Reject $a$.
\EndIf
\EndFor
\State Accept $a$. 
 \end{algorithmic}
\end{algorithm}

\begin{lemma}\label{lem: optimal tester for lines}
     Given a function $a:[0:m] \to [0:d] \times [0,dm]$, sample access to a distribution $\mathcal{D}'$ over $[0:m]$, and $0<\epsilon\leq1$, \Cref{algo:best_algo} is a non-adaptive one-sided tester with query complexity $\Omic(\log(m)/\epsilon)$ to test if $a$ is $\DSIC$ or $d^{\mathcal{D}'}_\DSIC(a) \geq \epsilon$. The tester uses $\lceil2/\epsilon\rceil$ samples and runs in time $\Tilde{\Omic}(n/\epsilon)+d^{\Omic(m)}/\epsilon$.
\end{lemma}
\begin{proof}{Proof}
    The samples may be drawn before any oracle response is seen, and each root-to-leaf query path depends only on its sampled $b$. Thus the tester is non-adaptive. It calls \Cref{algo:empty_tester} $\lceil2/\epsilon\rceil$ times, so the total number of queries is $\Omic(\log(m)/\epsilon)$ and the running time is $\Tilde{\Omic}(n/\epsilon)+d^{\Omic(m)}/\epsilon$.
    
    From \Cref{theorem:tree_DSIC} we know that if $a$ is $\DSIC$, for all $b \in [0:m]$ we have that $\mathcal{P}_{[b:b]} \neq \emptyset$. So we never reject $a$ if it $\DSIC$. From \Cref{theorem:tree_not_DSIC}, we know that if $d^{\mathcal{D}'}_\DSIC(a) \geq \epsilon$, then $\Pr_{b \sim \mathcal{D}'}[\mathcal{P}_{[b:b]}=\emptyset] \geq \epsilon$. Hence the tester fails to reject $a$ with probability at most $(1-\epsilon)^{\lceil2/\epsilon\rceil} \leq e^{-2}<1/3$.\Halmos
\end{proof}

The tester is non-adaptive and improves the query complexity of \Cref{algo:mid_tester}, at the cost of running time exponential in $m$.

\begin{proof}{Proof of \Cref{thm:ub_amazing}}

    The algorithm is obtained by taking \Cref{algo:newpricetester} and replacing calls to  \Cref{algo:mid_tester} by calls to \Cref{algo:best_algo}. Recall that the correctness, query and sampling complexity of \Cref{algo:best_algo} were characterized in  \Cref{lem: optimal tester for lines}.

    The proof of correctness is the same as the proof of \Cref{thm:ub}, since \Cref{lemma: line DSIC tester richer alloc} and \Cref{lem: optimal tester for lines} show \Cref{algo:mid_tester} and \Cref{algo:best_algo} have the same correctness guarantees.     
    
    A call to \Cref{algo:best_algo} with distance parameter $2^{-j}$ requires $\Omic(\log(m)2^j)$ queries. Thus

    \begin{align*}
       q&= \sum_{j=1}^\ell \Omic\left(\frac{\ell}{2^j\epsilon'}\right)
              \Omic\left(\log(m)2^j\right) \\
        &= \Omic\left(\frac{\log(m)\ell^2}{\epsilon'}\right) \\
        &= \Omic\left(\log(m)\frac{n}{\epsilon}
              \log^2\left(\frac{2n}{\epsilon}\right)\right)\,.
    \end{align*}

     \Cref{algo:best_algo} and \Cref{algo:mid_tester} have the same sampling complexity. Hence, as before, the tester uses $\Omic((n/\epsilon)\log^2(2n/\epsilon))$ conditional samples.

     Under the unit-cost convention of \Cref{thm:ub}, a call to $T$ with distance parameter $2^{-j}$ takes time $\Tilde{\Omic}(2^jn)+2^jd^{\Omic(m)}$. Therefore the running time is
    \begin{align*}
    \sum_{j=1}^\ell \Omic\left(\frac{\ell}{2^j\epsilon'}\right)
       \left(\Tilde{\Omic}(2^jn)+2^jd^{\Omic(m)}\right)
    &= \frac{\ell^2}{\epsilon'}
       \left(\Tilde{\Omic}(n)+d^{\Omic(m)}\right) \\
    &= \Tilde{\Omic}\left(\frac{n^2}{\epsilon}\right)
       +\Tilde{\Omic}\left(\frac{n}{\epsilon}\right)d^{\Omic(m)}.
    \end{align*}
    \Halmos

\end{proof}

\end{document}